\documentclass[leqno,a4paper,twoside,11pt]{article}
\usepackage[nointlimits,nosumlimits]{amsmath}
\usepackage{amsfonts,amssymb,amsthm,ifthen}
\usepackage{color}
\usepackage[arrow,matrix,curve]{xy}
\usepackage{hyperref}
\usepackage[utf8x]{inputenc}
\usepackage[square,numbers]{natbib}

\newtheorem{Thm}{Theorem}[section]
\newtheorem{Prp}[Thm]{Proposition}
\newtheorem{Lem}[Thm]{Lemma}
\newtheorem{Cor}[Thm]{Corollary}
\newtheorem{Def}[Thm]{Definition}

\newtheorem{Exp}[Thm]{Example}

\newtheorem{IntroThm}{Theorem}[section]

\newtheorem{IntroPrp}[IntroThm]{Proposition}

\theoremstyle{definition}

\newcommand{\mE}{\mathcal E}
\newcommand{\mF}{\mathcal F}
\newcommand{\mO}{\mathcal O}
\newcommand{\mP}{\mathcal P}
\newcommand{\mT}{\mathcal T}

\newcommand{\bN}{\mathbb N}
\newcommand{\bR}{\mathbb R}
\newcommand{\bZ}{\mathbb Z}

\newcommand{\Der}{\mathrm{Der}}
\newcommand{\ev}{\mathrm{ev}}
\newcommand{\End}{\mathrm{End}}
\newcommand{\Hom}{\mathrm{Hom}}
\newcommand{\hol}{\mathrm{hol}}
\newcommand{\Hol}{\mathrm{Hol}}
\newcommand{\rk}{\mathrm{rk}}
\newcommand{\pr}{\mathrm{pr}}
\newcommand{\spann}{\mathrm{span}}

\newcommand{\oA}{\overline{A}}
\newcommand{\oB}{\overline{B}}
\newcommand{\fgl}{\mathfrak{gl}}

\newcommand{\arr}[2]{%
  \begin{array}{@{}#1@{}}#2\end{array}}
\newcommand{\abs}[1]{\left| #1 \right|}
\newcommand*{\bigtimes}{\mathop{\raisebox{-.5ex}{\hbox{\huge{$\times$}}}}}
\newcommand{\dd}[2][]{\frac{\partial #1}{\partial #2}}
\newcommand{\scal}[3][]{\ifthenelse{\equal{#1}{}}{
  \left\langle #2,\,#3 \right\rangle
}{\ifthenelse{\equal{#1}{(}}{
  \left( #2,\,#3 \right)
}{\ifthenelse{\equal{#1}{[}}{
  \left[ #2,\,#3 \right]
}{
  #1\left( #2,\,#3 \right)
}}}}
\newcommand{\setsep}{\;\big|\;}

\newcommand\blfootnote[1]{%
  \begingroup
  \renewcommand\thefootnote{}\footnotetext{#1}%
  \endgroup
}

\renewcommand{\title}[1]{\vbox{\center\LARGE{\textsc{#1}}}\vspace{5mm}}
\renewcommand{\author}[1]{\vbox{\center\large{\textsc{#1}}}\vspace{5mm}}
\newcommand{\address}[1]{\vbox{\center\em#1}}
\newcommand{\email}[1]{\vbox{\center\tt#1}\vspace{5mm}}

\begin{document}

\title{The Twofold Way of Super Holonomy}

\author{Josua Groeger$^1$}

\blfootnote{Research funded by the Institutional Strategy of the University
of Cologne in the German Excellence Initiative.}

\address{Universit\"at zu K\"oln, Institut f\"ur Theoretische Physik,\\
  Z\"ulpicher Str. 77, 50937 K\"oln, Germany }

\email{$^1$groegerj@thp.uni-koeln.de}

\begin{abstract}
\noindent
There are two different notions of holonomy in supergeometry,
the supergroup introduced by Galaev and
our functorial approach motivated by super Wilson loops.
Either theory comes with its own version of
invariance of vectors and subspaces under holonomy.
By our first main result, the Twofold Theorem, these definitions are equivalent.
Our proof is based on the Comparison Theorem, our second main result,
which characterises Galaev's holonomy algebra as an algebra
of coefficients, building on previous results.
As an application, we generalise some of Galaev's results to
S-points, utilising the holonomy functor.
We obtain, in particular, a de Rham-Wu decomposition theorem
for semi-Riemannian S-supermanifolds.
\end{abstract}


\section{Introduction}

A connection on a vector bundle over a smooth manifold
gives rise to an isomorphism of fibres through parallel transport
along a connecting path. The group of such isomorphisms with respect to
loops, all starting and ending at the same point, is known as the holonomy group.
It is an important concept of an algebraic encoding of geometric properties
\cite{KN96,Joy00,MS99,GL10}.

The generalisation of holonomy to supergeometry is nontrivial due to the lack of
a sufficiently powerful notion of parallel transport in that context.
Recently, two approaches have been introduced that both overcome
this difficulty. In the first one due to Galaev \cite{Gal09},
a suitable generalisation of the Ambrose-Singer theorem is taken as the
definition of a super Lie algebra, which is then endowed to a super
Harish-Chandra pair, thus obtaining a super Lie group for every topological
point of the manifold.
In the second approach \cite{Gro14}, auxiliary Graßmann generators
are introduced that allow for a supergeometric parallel transport
modelling super Wilson loops \cite{MS10,BKS12}. The holonomy of an
$S$-point $x$ is then a Lie group valued functor $T\mapsto\Hol_x(T)$.

The relation between both theories is nontrivial, as to be elaborated henceforth.
The Twofold Theorem, to be stated next in an informal way,
is our main result for the sake of applications.

\begin{IntroThm}[Twofold Theorem]
\label{thmIntroTwofold}
A vector in the pullback of the super vector bundle under an $S$-point $x$
is invariant under Galaev's holonomy supergroup if and only if it is
invariant under $\Hol_x(T)$,
for a sufficiently large choice $T$ of auxiliary Graßmann generators.
An analogous statement holds for invariant subspaces.
\end{IntroThm}

In \cite{Gro14}, we argued that the generators of Galaev's holonomy algebra
can be extracted as certain coefficients from the Lie algebras occurring
in the functorial approach in the common situation of a topological point,
and a partial sketch of a proof was given.
In this article, we establish this observation, generalised to $S$-points,
as the Comparison Theorem,
which is our main result from a technical point of view.

\begin{IntroThm}[Comparison Theorem]
\label{thmIntroComparison}
Galaev's holonomy superalgebra can be characterised as the algebra of
$T$-coefficient matrices of the Lie algebras $\hol_x(T)$, for all $T$.
\end{IntroThm}

Our proof uses a formula for an odd derivative of parallel
transport interpreted as a homotopy and relates the pullback
of higher covariant derivatives to covariant derivatives of the pullback.
The Comparison Theorem is the main ingredient in the proof of the
Twofold Theorem. It further forms the basis of the following result.

\begin{IntroPrp}
\label{prpIntroSubfunctor}
The functor of points of Galaev's holonomy supergroup is
the smallest representable group functor
which contains $\Hol_x$ as a subfunctor.
On the level of Lie algebras,
the $T$-coefficients of monomials of sufficiently large degree agree.
\end{IntroPrp}

In view of our results mentioned so far, one might conjecture that
the functor of points of Galaev's holonomy supergroup should be the
sheafification of our supergroup functor, with respect to a Grothendieck
topology of submersions natural in our context.
This topology agrees with the fppf topology \cite{Zub09,MZ11}
on the category of Graßmann algebras.
We establish the following.

\begin{IntroPrp}
\label{prpIntroFPPF}
Both $\Hol_x$ and the functor of points of Galaev's holonomy supergroup
are sheaves in the fppf topology.
\end{IntroPrp}

Prp. \ref{prpIntroSubfunctor} and Prp. \ref{prpIntroFPPF} show
in a very precise way how the two approaches to super holonomy are related.
In the framework of either theory, it is natural to formulate
generalisations of the milestones of classical holonomy, that is to say
the holonomy principle, the theorem on parallel subbundles and
the de Rham-Wu theorem. By our Twofold Theorem, these generalisations
are, respectively, equivalent. While the holonomy principle
holds for general $S$-points (Thm. 2 in \cite{Gro14}),
the other two theorems mentioned were proved by Galaev in \cite{Gal09}
in the supergroup approach for topological points only.
We establish the general case in the functorial approach.

\begin{IntroThm}
\label{thmIntroParallel}
Parallel subbundles uniquely correspond to holonomy invariant subspaces.
\end{IntroThm}

\begin{IntroThm}[De Rham-Wu]
\label{thmIntroDeRhamWu}
A semi-Riemannian $S$-supermanifold splits into a product such that
the factors have weakly irreducible holonomy.
\end{IntroThm}

This article is organised as follows. In Sec. \ref{secTwofold}, we
briefly recall the relevant background and provide precise formulations
of Thm. \ref{thmIntroTwofold}, Thm. \ref{thmIntroComparison} and
Prp. \ref{prpIntroSubfunctor}. The proof of Thm. \ref{thmIntroComparison},
the Comparison Theorem, is deferred to the first part of
Sec. \ref{secCoefficientAlgebra}. In the second, we prove
Prp. \ref{prpIntroFPPF}. Thm. \ref{thmIntroParallel} and Thm. \ref{thmIntroDeRhamWu}
are the subject matter of Sec. \ref{secApplications}.
In a separate appendix, Sec. \ref{secFPPF}, we establish the aforementioned
characterisation of the fppf topology as a topology of submersions.
While this can be deduced from a result by Schmitt in \cite{Sch89},
we provide an independent proof based on a less abstract result
by Esin and Ko\c{c} in \cite{EK07}.

\section{The Twofold Theorem}
\label{secTwofold}

In this section, we first recall the functorial holonomy theory of \cite{Gro14} and
introduce a slight generalisation of Galaev's approach developed in \cite{Gal09}.
Having set the stage, we precisely formulate and prove
our main results concerning the relation
between both theories. Thm. \ref{thmIntroTwofold} and Thm. \ref{thmIntroComparison}
correspond to Thm. \ref{thmTwofold} and Thm. \ref{thmGalaevFunctor}, respectively,
while Prp. \ref{prpIntroSubfunctor} is split into Prp. \ref{prpGalaevFunctor},
Prp. \ref{prpSufficientlyLarge} and Prp. \ref{prpSmallestRepresentable}.

Let $M=(M_{\overline{0}},\mO_M)$ be a supermanifold in the sense of
Berezin-Kostant-Leites (\cite{Lei80}, \cite{Var04}, \cite{CCF11})
with underlying classical manifold $M_{\overline{0}}$. Concerning
notation, we shall use the subscript $\overline{0}$ also to denote
the even part of a super vector space, and $\overline{1}$ to denote
the odd part. Let $\mE$ be a super vector bundle
over $M$ considered as a sheaf of locally free $\mO_M$ supermodules, such as
the tangent sheaf $\mT M:=\Der(\mO_M)$. Moreover, we fix a superpoint
$S:=\bR^{0|L}$, an $S$-connection $\nabla$ on the sheaf
$\mE_S:=\mE\otimes_{\mO_M}\mO_{S\times M}$ and an $S$-point $x:S\rightarrow M$
on $M$.

As detailed in \cite{Gro14}, $\nabla$ gives rise to parallel transport
operators $P_{\gamma}:x^*\mE\rightarrow y^*\mE$ along an $S$-path
$\gamma:S\times[0,1]\rightarrow M$ connecting $x$ with another $S$-point $y$.
The holonomy group $\Hol_x$ is defined as the
set of parallel transports $P_{\gamma}$ such that $\gamma$ is a piecewise
smooth $S$-loop starting and ending in $x$.
For notational simplicity, we shall not denote the dependence on
$\nabla$ explicitly in the following. $\Hol_x$ carries the structure of
a Lie group. By a theorem of Ambrose-Singer type, its Lie algebra $\hol_x$ is
generated by endomorphisms of the form
\begin{align}
\label{eqnHolonomyAlgebra}
\{P_{\gamma}^{-1}\circ\scal[R_y]{u}{v}\circ P_{\gamma}\setsep
y:S\rightarrow M\,,\;\gamma:x\rightarrow y\;\;\mathrm{pw. smooth}\,,\;
u,v\in(y^*\mT M)_{\overline{0}}\}
\end{align}
where $R$ denotes the curvature tensor with respect to $\nabla$.
As it stands, the holonomy group $\Hol_x$ contains only a limited
amount of information, making it necessary
to consider a larger set of loops.
To that end, let $T=\bR^{0|L'}$ be another superpoint and consider $x$
as an $S\times T$-point, denoted $x_T:S\times T\rightarrow M$.
The prescription $T\mapsto\Hol_x(T):=\Hol_{x_T}$ extends to
a Lie group valued functor, referred to as the holonomy group functor $\Hol_x$
in the following. Similarly, the assignment of Lie algebras
$T\mapsto\hol_x(T)$ establishes a functor, denoted $\hol_x$, in its own right.
Both $\Hol_x$ and $\hol_x$ can be considered as functors from the category
$\mP$ of superpoints $T=\bR^{0|L'}$ to the category $\mathrm{Sets}$.

Galaev's holonomy supergroup studied in \cite{Gal09} is defined
for a classical point $x\in M_{\overline{0}}$, that is an $S$-point
with $L=0$ in the notation above. For comparison with the holonomy group
functor, we introduce a slight generalisation of that theory next.
To begin with, we define higher covariant derivatives of tensors of
curvature type.
Choose an auxiliary $S$-connection on $\mT M$, which will also be
referred to as $\nabla$ and left implicit.

\begin{Def}
\label{defHigherCovariantDerivatives}
Let $F\in\Hom_{\mO_{S\times M}}(\mT M_S\otimes_{\mO_{S\times M}}\mT M_S\otimes_{\mO_{S\times M}}\mE_S,\mE_S)_{\overline{0}}$
be a tensor and $Y_1,\ldots,Y_{k+1}\in\mT M_S$.
For $X,Y\in\mT M_S$, we define
\begin{align*}
\scal[(\nabla_{Y_1}F)]{X}{Y}&:=\nabla_{Y_1}\circ\scal[F]{X}{Y}
-\scal[F]{\nabla_{Y_1}X}{Y}-(-1)^{\abs{Y_1}\abs{X}}\scal[F]{X}{\nabla_{Y^1}Y}\\
&\qquad-(-1)^{\abs{Y_1}(\abs{X}+\abs{Y})}\scal[F]{X}{Y}\circ\nabla_{Y_1}
\end{align*}
and, recursively,
\begin{align*}
\nabla^{k+1}_{Y_{k+1},\ldots,Y_1}F
&:=\nabla_{Y_{k+1}}(\nabla^k_{Y_k,\ldots,Y_1}F)\\
&\qquad-\sum_j(-1)^{\abs{Y_{k+1}}(\abs{Y_k}+\ldots+\abs{Y_1})}\nabla^k_{Y_k,\ldots,Y_{j+1},\nabla_{Y_{k+1}}Y_j,Y_{j-1},\ldots,\ldots,Y_1}F
\end{align*}
\end{Def}

The covariant derivatives $\nabla^kF$ are tensors and, as such, may be pulled back
to $y:S\times T\rightarrow M$. As in \cite{Gro14} we write,
by a slight abuse of notation,
\begin{align*}
(\nabla^k_{Y_k,\dots,Y_1}F)_y(u,v):=((\nabla^kF)_y)_{Y_k,\ldots,Y_1}(u,v)
\end{align*}
for $Y_r,\ldots,Y_1,u,v\in y^*\mT M$.

\begin{Def}
\label{defGalaevAlgebra}
Let $x:S\rightarrow M$ be an $S$-point. Let $\hol_x^{\mathrm{Gal}}$ denote
the ($\bR$-)super Lie algebra generated by operators
\begin{align*}
P_{\gamma}^{-1}\circ\left(\nabla^k_{Y_k,\ldots,Y_1}R\right)_y(u,v)\circ P_{\gamma}\subseteq\End_{\mO_S}(x^*\mE)\cong\fgl_{\rk\mE}(\mO_S)
\end{align*}
with arbitrary $S$-point $y:S\rightarrow M$,
$S$-path $\gamma:x\rightarrow y$ and $Y_1,\ldots,Y_r,u,v\in y^*\mT M$.
\end{Def}

The reader should note that Def. \ref{defGalaevAlgebra},
like Def. \ref{defHigherCovariantDerivatives},
implicitly depends on the choice of an auxiliary connection
on $\mT M$. As done by Galaev for $L=0$, one can directly show that,
in fact, $\hol_x^{\mathrm{Gal}}$ is independent thereof. We will
omit this proof here, since that statement is a direct corollary of
Thm. \ref{thmGalaevFunctor} below.

The even part
$(\hol_x^{\mathrm{Gal}})_{\overline{0}}\subseteq\fgl_{\rk\mE}(\mO_S)_{\overline{0}}$
is the Lie algebra of a unique immersed connected Lie subgroup of
$GL_{\rk\mE}(\mO_S)$ (see Chp. 2 of \cite{GOV97}), that we will refer to
as $(\Hol_x^{\mathrm{Gal}})_{\overline{0}}^0$.
Moreover, we define the Lie group
\begin{align*}
(\Hol_x)_{\overline{0}}:=\Hol_x(T=\bR^{0|0})=\{P_{\gamma}\setsep\gamma:S\times[0,1]\rightarrow M\,,\;\gamma:x\rightarrow x\}\subseteq GL_{\rk\mE}(\mO_S)
\end{align*}
whose Lie algebra is contained in $(\hol_x^{\mathrm{Gal}})_{\overline{0}}$,
and which is not connected, in general.
Now let $(\Hol_x^{\mathrm{Gal}})_{\overline{0}}\subseteq GL_{\rk\mE}(\mO_S)$
be the Lie group generated
by $(\Hol_x^{\mathrm{Gal}})_{\overline{0}}^0$ and $(\Hol_x)_{\overline{0}}$,
which comes with the natural adjoint action on
$\hol_x^{\mathrm{Gal}}\subseteq\fgl_{\rk\mE}(\mO_S)$.
We thus obtain a super Harish-Chandra pair corresponding to a super Lie group.

\begin{Def}
Galaev's holonomy supergroup $\Hol_x^{\mathrm{Gal}}$ is the super Lie group
determined by the super Harish-Chandra pair
$((\Hol^{\mathrm{Gal}}_x)_{\overline{0}},\hol_x^{\mathrm{Gal}})$.
\end{Def}

We shall next define the holonomy algebra of coefficients
of $\hol_x(T)$ with respect to generators of $\mO_T$.
Continuing the above convention, the number $L$ will always
denote the odd dimension of $S=\bR^{0|L}$, whereas $L'$ is as in
$T=\bR^{0|L'}$.

For a fixed number $L'>0$, we consider the generators $\eta_{L'}^1,\ldots,\eta_{L'}^{L'}$
of the Graßmann algebra $\mO_T=\bigwedge\bR^{L'}$ corresponding to the standard
basis of $\bR^{L'}$. The canonical inclusion $\bR^{L'}\subseteq\bR^{L''}$
for $L''>L'$ induces an identification of the generators
$\eta_{L'}^i$ and $\eta_{L''}^i$ (for $1\leq i\leq L'$).
In the following, we will simply write
$\eta^i$ as an element of $\bigwedge\bR^{L''}$ for some $L''\geq L'$
whose exact value is not important unless stated otherwise.

\begin{Def}
\label{defCoefficientHolonomy}
We define the coefficient holonomy algebra $\hol^{\mathrm{C}}_x$
to be the set of coefficient matrices of $T$-generators as follows.
\begin{align*}
&\hol^{\mathrm{C}}_x:=(\hol^{\mathrm{C}}_x)_{\overline{0}}\oplus
(\hol^{\mathrm{C}}_x)_{\overline{1}}\subseteq\fgl_{\rk\mE}(\mO_S)\\
&(\hol^{\mathrm{C}}_x)_{\overline{i}}:=\{h\setsep\exists L'\in\bN\,,\;
\exists A=\sum_IA^I\cdot\eta^I\in\hol_x(T=\bigwedge\bR^{L'})\,,\;
\exists I\,:\;\abs{I}=\overline{i}\,,\;h=A^I\}
\end{align*}
\end{Def}

\begin{Lem}
\label{lemSupermoduleAlgebra}
$\hol^{\mathrm{C}}_x$ is an $\mO_S$-supermodule as well as a super Lie algebra,
with operations induced from $\fgl_{\rk\mE}(\mO_S)$.
\end{Lem}

\begin{proof}
The sum of two elements $X,Y\in\hol^{\mathrm{C}}_x$ is contained in the same
set. In case of opposite parity, this is by definition. Otherwise,
it follows from the fact that, by (\ref{eqnHolonomyAlgebra}),
$A\cdot\eta^K\in\hol_x(T)$ for every $A\in\hol_x(T)$ and every even monomial
$\eta^K\in\mO_T$. A similar argument concludes the proof that
$\hol^{\mathrm{C}}_x$ is an $\mO_S$-supermodule.
Finally, we claim that $\scal[[]{X}{Y}\in\hol^{\mathrm{C}}_x$.
Let $A=A^I\eta^I\in\hol_x(T)$, $B=B^J\eta^J\in\hol_x(T')$,
and let $I_0$ and $J_0$ denote
multiindices such that $X=A^{I_0}$ and $Y=B^{J_0}$.
Without loss of generality, we may assume $T=T'$. Moreover,
we may alter $B$ such as to achieve $I\cdot J\neq 0$ for all $I$ and $J$
occurring. In terms of the generators (\ref{eqnHolonomyAlgebra}),
this means to change the instances of $\gamma$, $u$ and $v$ accordingly.
It then follows that
$\scal[[]{X}{Y}=\scal[[]{A^{I_0}}{B^{J_0}}=\scal[[]{A}{B}^{I_0J_0}$,
thus concluding the proof.
\end{proof}

A Lie supergroup over $S$ is a group object in the category of
supermanifolds over $S$ and, as such, possesses an $\mO_S$-Lie
superalgebra (sheaf) which is (locally) free. This is detailed
in a forthcoming article by Alldridge and Coulembier. In the present
setting, it is thus natural to conjecture
that $\hol^{\mathrm{C}}_x$ should be free as an $\mO_S$-supermodule.
However, the following example, which resembles Exp. 4 of \cite{Gro14},
shows that this conjecture is false, in general.

\begin{Exp}
Let $M=\bR^{0|1}=(*,\langle\theta\rangle)$, $\mE=\mT M$ and $S=\bR^{0|L}$ with some
$L\geq 2$. Denoting the standard $S$-coordinates by $\hat{\eta}^j$, we define
an $S$-connection on $\mE_S$ by prescribing
$\nabla_{\partial_{\theta}}\partial_{\theta}=\hat{\eta}^1\hat{\eta}^2\theta\partial_{\theta}$.
A short calculation shows that
\begin{align*}
P_{\gamma}^{-1}\circ\scal[R_y]{u}{v}P_{\gamma}[w]=-2\hat{\eta}^1\hat{\eta}^2u^{\theta}v^{\theta}\cdot w
\end{align*}
writing $u=(y^*\partial_{\theta})\cdot u^{\theta}\in(y^*\mE)_{\overline{0}}$
and analogously for $v$.
Any element $C\in\hol^C_*$ is, therefore, of the form $C=\hat{\eta}^1\hat{\eta}^2\cdot\tilde{C}$.
\end{Exp}

\begin{Thm}[Comparison Theorem]
\label{thmGalaevFunctor}
The coefficient holonomy algebra is
Galaev's holonomy superalgebra $\hol_x^\mathrm{C}=\hol_x^{\mathrm{Gal}}$.
\end{Thm}

Technically, the Comparison Theorem is the main result of this article.
A partial proof of one implication, in the case $S=\bR^{0|0}$,
was already given in Sec. 4 of \cite{Gro14}.
This was done by suitably expressing covariant
derivatives up to second order by infinitesimal parallel transport and,
moreover, considering special $S\times T$-points of $M$ ((20) in that reference).
We are convinced that the argument can be generalised in that,
by some inductive proof, higher covariant derivatives to any order
should be expressible by means of parallel transport.
While such a statement would certainly be of independent interest,
already calculations at low orders involve some technical complexity.

For this reason, we shall provide a different proof, which we defer
to Sec. \ref{secCoefficientAlgebra} below.
It is based on a formula for an odd derivative of parallel transport, when
interpreted as a corresponding homotopy, and a precise study of the
different pullback derivatives involved.

For further comparison of the two theories of super holonomy,
it is most natural to
relate the functors $\Hol_x$ and $\hol_x$ to the functors of points
of $\Hol_x^{\mathrm{Gal}}$ and $\hol_x^{\mathrm{Gal}}$, respectively.
It will be helpful to consider all functors occurring as subfunctors of the functor
of points of $\fgl_{\rk\mE}(\mO_S)$ to be described next.

Consider $\fgl_{n|m}(\mO_S)$ as a real super vector space, with
$\bZ_2$-grading induced by the supermatrix grading and the natural
grading on $\mO_S$. With respect to a choice of generators
$\hat{\eta}^1,\ldots,\hat{\eta}^L$ of $\mO_S$, a natural basis is given by
the matrices $(E^{lmI})$ with entries
$(E^{lmI})_{kn}:=\delta^l_k\delta^m_n\cdot\hat{\eta}^I$, where
$1\leq l,m\leq n+m$, $I$ is a multiindex, and the parity of $E^{lmI}$
is given by $\abs{l}+\abs{m}+\abs{I}$.
Let $\fgl_{n|m}^S$ denote the supermanifold corresponding to that
super vector space.
Moreover, let $GL_{n|m}^S$ denote the open subsupermanifold
such that the points in the underlying manifold
$(\fgl_{n|m}^S)_{\overline{0}}$ are invertible as supermatrices.
It is known that this condition is equivalent to invertibility of the matrix
arising from projecting away the $\mO_S$-generators.

\begin{Lem}
\label{lemGLS}
The functors of points satisfy
$\fgl_{n|m}^S(T)\cong\fgl_{n|m}(\mO_{S\times T})_{\overline{0}}$ (as Lie algebras)
and $GL_{n|m}^S(T)\cong GL_{n|m}(\mO_{S\times T})$ (as Lie groups).
Moreover, the super Harish-Chandra pair of $GL_{n|m}^S$ is given by
$(GL_{n|m}(\mO_S),\fgl_{n|m}(\mO_S))$ together with the usual adjoint action.
\end{Lem}

\begin{proof}
The first assertion is clear by the general formula
$V(T)\cong(V\otimes\mO_T)_{\overline{0}}$, where $V$ on the left hand side
is the supermanifold associated to a super vector space (or super Lie algebra),
also denoted by $V$ on the right hand side. The corresponding statement
for the Lie groups holds by the characterisation of invertibility
of a supermatrix mentioned above.
With these identifications, the Harish-Chandra characterisation
is also established.
\end{proof}

While the first part of Lem. \ref{lemGLS} is valid for $T$ in the category
$\mathrm{SMan}$ of all supermanifolds,
we will use this result for the subcategory $\mP$ of superpoints $T=\bR^{0|L'}$.

Let $F_1$ and $F_2$ be functors to the category $\mathrm{Sets}$.
$F_1$ is said to be a subfunctor of $F_2$ if
$F_1(T)\subseteq F_2(T)$ for every object $T$,
and for every morphism $\varphi:T\rightarrow S$,
the morphism $F_1(\varphi)$ arises by restriction from $F_2(\varphi)$.
This is true for the functors arising in the context of holonomy as
summarised in the following lemma.

\begin{Lem}
\label{lemSubfunctors}
The functors $\hol_x$ and $\hol_x^{\mathrm{Gal}}$ are subfunctors of
(the functor of points of) $\fgl_{\rk\mE}^S$.
Similarly, $\Hol_x$ and $\Hol_x^{\mathrm{Gal}}$ are subfunctors of
$GL_{\rk\mE}^S$ (and, as such, also subfunctors of $\fgl_{\rk\mE}^S$).
\end{Lem}

As the first corollary of the Comparison Theorem, Thm. \ref{thmGalaevFunctor},
we establish important inclusions.

\begin{Prp}
\label{prpGalaevFunctor}
For every $T=\bigwedge\bR^L$, there are canonical inclusions
\begin{align*}
\hol_x(T)\subseteq
(\hol_x^{\mathrm{Gal}}\otimes\mO_T)_{\overline{0}}=\hol_x^{\mathrm{Gal}}(T)
\;,\quad
\Hol_x(T)\subseteq\Hol_x^{\mathrm{Gal}}(T)
\end{align*}
of Lie algebras and Lie groups, respectively, which are functorial in $T$
and thus induce natural transformations.
\end{Prp}

\begin{proof}
The first inclusion is a direct corollary of Thm. \ref{thmGalaevFunctor}.

$\hol_x(T)$ and $\hol_x^{\mathrm{Gal}}(T)$ are both Lie subalgebras
of $\fgl_{\rk\mE}(\mO_{S\times T})$.
By standard Lie group theory (cf. Chp. 2 of \cite{GOV97}),
the connected components of the corresponding holonomy groups thus satisfy
\begin{align*}
\Hol_x(T)^0\subseteq\Hol_x^{\mathrm{Gal}}(T)^0\subseteq
GL_{\rk\mE}(\mO_{S\times T})
\end{align*}
Every element $P_{\gamma}\in\Hol_x(T)$ comes with an $S\times T$-path
$\gamma:S\times T\times[0,1]\rightarrow M$, which may be homotoped to
the underlying $S$-path $\tilde{\gamma}$ obtained from $\gamma$ by
the standard inclusion $S\rightarrow S\times T$.
This homotopy gives rise to a path $P_{\tilde{\gamma}}\rightarrow P_{\gamma}$.
In particular, the composition
$P:=P_{\gamma}\cdot P_{\tilde{\gamma}}^{-1}$ is connected to the identity
and, therefore, contained in $\Hol_x^{\mathrm{Gal}}(T)^0$.
Since, moreover,
$P_{\tilde{\gamma}}\in(\Hol_x)_{\overline{0}}\subseteq\Hol_x^{\mathrm{Gal}}(T)$,
it follows that $P_{\gamma}=P\cdot P_{\tilde{\gamma}}\in\Hol_x^{\mathrm{Gal}}(T)$,
which proves the second inclusion.

The various pullbacks with respect to a morphism $\varphi:T'\rightarrow T$
are all standard, which shows functoriality.
\end{proof}

The inclusion of Prp. \ref{prpGalaevFunctor} is, in general, proper,
as can be seen e.g. in Exp. 4 of \cite{Gro14}. Nevertheless, the holonomy
group functor $\Hol_x(T)$ still contains all relevant information,
to be detailed henceforth.

Reconsider Def. \ref{defCoefficientHolonomy} of $\hol_x^{\mathrm{C}}$
as the set of coefficients with respect to arbitrary $T=\bR^{0|L}$:
From the form (\ref{eqnHolonomyAlgebra}), it is clear that small
$T$ are irrelevant in the sense that $A\in\hol_x(T)$ implies
$A\in\hol_x(T')$ for $T'=\bR^{0|L'}$ with $L'\geq L$.
Therefore, coefficients with respect to $T$ are coefficients with respect
to $T'$. Since $\fgl(n|m,\mO_S)$ is finite-dimensional,
this inclusion stabilises.
For any fixed sufficiently large $T$, we thus obtain
\begin{align}
\label{eqnSufficientlyLarge}
(\hol^{\mathrm{C}}_x)_{\overline{i}}:=\{h\setsep
\exists A=\sum_IA^I\cdot\eta^I\in\hol_x(T)\,,\;
\exists I\,:\;\abs{I}=\overline{i}\,,\;h=A^I\}
\end{align}
Similarly, it is clear that $A\cdot\eta^J\in\hol_x(T)$ for $A\in\hol_x(T)$
and an even monomial $\eta^J\in\mO_T$. It follows that all elements
of $\hol^{\mathrm{C}}_x$ occur as coefficients of monomials of sufficiently
large degree. This observation can be formulated in the following form.

\begin{Prp}
\label{prpSufficientlyLarge}
Writing $\mO_T$ as a direct sum
$\mO_T=(\mO_T)_{\mathrm{deg}\leq N}\oplus(\mO_T)_{\mathrm{deg}>N}$
for some $N<L$, we obtain
\begin{align*}
\hol_x(T)=(\hol_x(T))_{\mathrm{deg}\leq N}\oplus(\hol_x^{\mathrm{Gal}}\otimes(\mO_T)_{\mathrm{deg}>N})_{\overline{0}}\subseteq(\hol_x^{\mathrm{Gal}}\otimes\mO_T)_{\overline{0}}
\end{align*}
provided that both $N$ and $L$ are sufficiently large.
\end{Prp}

By Prp. \ref{prpGalaevFunctor}, the functor $\Hol_x^{\mathrm{Gal}}$ contains
$\Hol_x$ as a subfunctor. Moreover it is, by definition, representable.
By the following result, it is the smallest one with these properties.

\begin{Prp}
\label{prpSmallestRepresentable}
$\Hol_x^{\mathrm{Gal}}$ is the smallest representable group functor
which contains $\Hol_x$ as a subfunctor.
\end{Prp}

\begin{proof}
Assume that $\widetilde{\Hol_x}$ is a super Lie group with super
Lie algebra $\widetilde{\hol_x}$ such that
\begin{align*}
\Hol_x(T)\subseteq\widetilde{\Hol_x}(T)\subseteq\Hol_x^{\mathrm{Gal}}(T)
\end{align*}
for every $T$, in a functorial way.
Considering Harish-Chandra pairs we will show that,
in this case, the super Lie groups $\widetilde{\Hol_x}=\Hol_x^{\mathrm{Gal}}$
agree.

By the corresponding inclusions for the Lie algebra valued
functors and Prp. \ref{prpSufficientlyLarge}, we deduce
\begin{align*}
(\hol_x^{\mathrm{Gal}}\otimes(\mO_T)_{\deg>N})_{\overline{0}}
\subseteq(\widetilde{\hol_x}\otimes(\mO_T)_{\deg>N})_{\overline{0}}
\subseteq(\hol_x^{\mathrm{Gal}}\otimes(\mO_T)_{\deg>N})_{\overline{0}}
\end{align*}
It follows that $\hol_x^{\mathrm{Gal}}=\widetilde{\hol_x}$ and, moreover,
the connected Lie groups of the respective even parts coincide.
In particular, $(\widetilde{\Hol_x})_{\overline{0}}$ contains
$(\Hol_x^{\mathrm{Gal}})_{\overline{0}}^0$.
By assumption, it also contains $(\Hol_x)_{\overline{0}}$.
The inclusion $(\Hol_x^{\mathrm{Gal}})_{\overline{0}}\subseteq(\widetilde{\Hol_x})_{\overline{0}}$ becomes immediate. Again by assumption, this is, in fact,
an equality.
\end{proof}

The non-trivial relation between the functors of both holonomy theories
is revealed in Prp. \ref{prpGalaevFunctor}, Prp. \ref{prpSufficientlyLarge}
and Prp. \ref{prpSmallestRepresentable}. We will further show that
both group functor are sheaves in the fppf topology, see Prp. \ref{prpHolSheaves}
below. In particular, $\Hol_x^{\mathrm{Gal}}$ is not the sheafification
of $\Hol_x$, as one might conjecture in light of the previous results.
When it comes to applications, however, both holonomy theories turn out
to be equivalent. This is detailed in the Twofold Theorem, our main result
in this regard.

\begin{Thm}[Twofold Theorem]
\label{thmTwofold}
Let $x:S\rightarrow M$ be an $S$-point and set $\mE_x=x^*\mE$.
Let $T$ be fixed sufficiently large as in (\ref{eqnSufficientlyLarge}).
\begin{enumerate}
\renewcommand{\labelenumi}{(\roman{enumi})}
\item Let $X_x\in\mE_x$. Then $(\Hol_x)_{\overline{0}}\cdot X_x=X_x$
and $\hol_x^{\mathrm{Gal}}\cdot X_x=0$ if and only if
$\Hol_x(T)\cdot X_x=X_x$.
\item Let $\mF_x\subseteq\mE_x$ be a free $\mO_S$-submodule. Then
$(\Hol_x)_{\overline{0}}\cdot\mF_x\subseteq\mF_x$ and
$\hol_x^{\mathrm{Gal}}\cdot\mF_x\subseteq\mF_x$ if and only if
$\Hol_x(T)\cdot\mF_x\subseteq\mF_x$.
\end{enumerate}
In both parts, the respective second condition for one (large) $T$
may be equivalently replaced by the corresponding condition for
\emph{every} $T$.
\end{Thm}

\begin{proof}
We show part (i) of the theorem and omit the analogous proof of (ii).

Starting with the second implication, assume that $X_x$ satisfies
$\Hol_x(T)\cdot X_x=X_x$.
Since $(\Hol_x)_{\overline{0}}\subseteq\Hol_x(T)$, it follows that
$X_x$ is preserved by this group. Moreover, passing to the Lie algebra
$\hol_x(T)$, it also follows that $\hol_x(T)\cdot X_x=0$.
Considering coefficients (\ref{eqnSufficientlyLarge}), we obtain
$\hol_x^C\cdot X_x=0$, provided that $T$ is sufficiently large.
The implication is then immediate thanks to Thm. \ref{thmGalaevFunctor}.

Conversely, assume that $(\Hol_x)_0\cdot X_x=X_x$
and $\hol_x^{\mathrm{Gal}}\cdot X_x=0$. Being connected, the Lie
group $(\Hol_x^{\mathrm{Gal}})_{\overline{0}}^0$ is generated by
$\exp(\hol_x^{\mathrm{Gal}})$ and thus preserves $X_x$. It follows
that $(\Hol_x^{\mathrm{Gal}})_{\overline{0}}\cdot X_x=X_x$.
Consider the super Lie group $G:=GL_{\rk\mE}^S$ acting naturally on $\mE_x$.
Denoting by $G_x$ the stabiliser super Lie subgroup of $G$ with repect to $X_x$,
the assumptions imply that the super Harish-Chandra pair
$((\Hol^{\mathrm{Gal}}_x)_{\overline{0}},\hol_x^{\mathrm{Gal}})$ of
$\Hol_x^{\mathrm{Gal}}$ is a subpair of the one corresponding to $G_x$.
Therefore, for each $T$,
$\Hol_x^{\mathrm{Gal}}(T)\subseteq G_x(T)$ (cf. Prp. 8.4.7 of
\cite{CCF11}), i.e. $\Hol_x^{\mathrm{Gal}}(T)\cdot X_x=X_x$.
By the inclusion of Prp. \ref{prpGalaevFunctor}, it then follows that
$\Hol_x(T)\cdot X_x=X_x$.
\end{proof}

\section{The Holonomy Functor and its Algebra of Coefficients}
\label{secCoefficientAlgebra}

This section contains a thorough anaysis of the holonomy functor
and forms the technical core of the present article.
We provide a proof of the Comparison Theorem, Thm. \ref{thmIntroComparison} and
Thm. \ref{thmGalaevFunctor} above,
according to the strategy mentioned in the context of its formulation.
Finally, we establish Prp. \ref{prpIntroFPPF} as
Prp. \ref{prpHolSheaves}, stating that all functors
occuring are sheaves in the fppf topology. The proof turns out to be
quite simple, provided that a suitable characterisation of the topology
is taken into account.

\subsection{Proof of the Comparison Theorem}

We begin with the following two lemmas, which relate the pullback of
higher covariant derivatives of a tensor $F\in\Hom_{\mO_{S\times M}}(\mT M_S\otimes_{\mO_{S\times M}}\mT M_S\otimes_{\mO_{S\times M}}\mE_S,\mE_S)_{\overline{0}}$
of curvature type as in Def. \ref{defHigherCovariantDerivatives}
to covariant pullback derivatives of the pullback
$F_y=y^*F$ with respect to a point $y:S\times T\rightarrow M$.
It will be sufficient to consider consecutive applications of
first-order covariant pullback derivatives rather than higher-order ones.
For convenience, we recall the definition (Def. 14 of \cite{Gro14}),
with $u,v\in y^*\mT M$ and $X\in\mT(S\times T)$.
\begin{align*}
(x^*\nabla)_X\scal[F_x]{u}{v}&:=(x^*\nabla)_X\circ\scal[F_x]{u}{v}-\scal[F_x]{(x^*\nabla)_X(u)}{v}\\
&\qquad-(-1)^{\abs{X}\abs{u}}\scal[F_x]{u}{(x^*\nabla)_X(v)}
-(-1)^{\abs{X}(\abs{u}+\abs{v})}\scal[F_x]{u}{v}\circ(x^*\nabla)_X
\end{align*}
The reader should note that this definition,
like Def. \ref{defHigherCovariantDerivatives}, depends
on the choice of an auxiliary connection on $\mT M$.
It will be used with $X=\dd{\eta^i}$ where, in all of the following,
the $\eta^i$ will denote the standard odd generators of $\mO_T$ (not of $\mO_S$),
as in the beginning of this section.

\begin{Lem}
\label{lemHolk}
Let $y:S\times T\rightarrow M$ be an $S\times T$-point.
Then every operator of the form
\begin{align*}
(y^*\nabla)_{\partial_{\eta^{i_k}}}\circ\ldots\circ(y^*\nabla)_{\partial_{\eta^{i_1}}}F_y
\end{align*}
can be written as an $\mO_{S\times T}$-linear combination of operators of the form
\begin{align*}
\left(\nabla^l_{\partial_{\xi^{j_l}},\ldots,\partial_{\xi^{j_1}}}F\right)_y
\end{align*}
with $0\leq l\leq k$ and $\xi=(x,\theta)$ coordinates around $y\in M_0$.
\end{Lem}

\begin{proof}
The base case $k=0$ is trivial. We assume, by induction, that the assumption
holds for $k$. To establish the case $k+1$, we unwind the definitions
to calculate
\begin{align*}
&(y^*\nabla)_{\partial_{\eta^{i_{k+1}}}}
\left(\nabla^l_{\partial_{\xi^{j_l}},\ldots,\partial_{\xi^{j_1}}}F\right)_y\\
&\qquad=\partial_{\eta^{i_{k+1}}}(y^*(\xi^m))\cdot
\left(\nabla_{\partial_{\xi^m}}
\nabla^l_{\partial_{\xi^{j_l}},\ldots,\partial_{\xi^{j_1}}}F\right)_y\\
&\qquad=\partial_{\eta^{i_{k+1}}}(y^*(\xi^m))\cdot
\left(\nabla^{l+1}_{\partial_{\xi^m},\partial_{\xi^{j_l}},\ldots,\partial_{\xi^{j_1}}}F\right.\\
&\qquad\qquad\qquad\qquad\qquad\qquad\left.+\sum_n(\pm 1)\cdot\nabla^l_{\partial_{\xi^{j_l}},\ldots,\partial_{\xi^{j_{n+1}}},\nabla_{\partial_{\xi^m}}\partial_{\xi^{j_n}},\partial_{\xi^{j_{n-1}}},\ldots,\partial_{\xi^{j_1}}}F\right)_y
\end{align*}
Pulling the Christoffel symbol in $\nabla_{\partial_{\xi^m}}\partial_{\xi^{j_n}}=\Gamma_{mj_n}^h\partial_{\xi^h}$ in front of the $l$-order covariant derivative, we end up with
an $\mO_{S\times T}$-linear combination of operators of the form
stated, with $0\leq l+1\leq k+1$ provided that $0\leq l\leq k$.
The case $k+1$ then follows immediately with the Leibniz rule.
\end{proof}

While Lem. \ref{lemHolk} holds in general, the following one will be
established for a suitable generalisation of the
special points (20) of \cite{Gro14} to be introduced next.
Given an $S$-point $q:S\rightarrow M$, consider the
following $S\times T$-point $y$, which is defined with respect to some
choice of coordinates $\xi=(x,\theta)$ of $M$ around $q_0\in M_{\overline{0}}$.
\begin{align}
\label{eqnGalaevYGeneralised}
y^{\sharp}(\theta^i):&=\eta^i+q^{\sharp}(\theta^i)&\in\,&(\mO_T)_{\overline{1}}+(\mO_S)_{\overline{1}}\\
y^{\sharp}(x^j):&=\sum_{n=0}^{k-1}\eta^{d_1+n(2d_0)+2j-1}\eta^{d_1+n(2d_0)+2j}+q^{\sharp}(x^j)&\in\,&(\mO_T)_{\overline{0}}+(\mO_S)_{\overline{0}}\nonumber
\end{align}
Here, $k\in\bN$ is a fixed number, and the indices $j$ and $i$
run through $1\leq j\leq d_0$ and $1\leq i\leq d_1$,
respectively, where we abbreviate $d_i:=(\dim M)_{\overline{i}}$
for the even ($i=0$) and odd ($i=1$) dimensions of $M$.
For this definition to make sense, it is implicitly understood that $T$
be sufficiently large.
The point $y$ is constructed such as to satisfy
\begin{align*}
y^*(\nabla_{\partial_{\theta^i}}Z)=(y^*\nabla)_{\partial_{\eta^i}}(y^*Z)\,,\quad
\eta^{d_1+n(2d_0)+2j}\cdot y^*(\nabla_{\partial_{x^j}}Z)=(y^*\nabla)_{\partial_{\eta^{d_1+n(2d_0)+2j-1}}}(y^*Z)
\end{align*}
and, similarly,
\begin{align}
\label{eqnGalaevTermsGeneralised}
y^*\left((\nabla_{\partial_{\theta^i}}F)(\partial_{\xi^a},\partial_{\xi^b})\right)
&=\left((y^*\nabla)_{\partial_{\eta^i}}F_y\right)(y^*\partial_{\xi^a},y^*\partial_{\xi^b})\\
\eta^{d_1+n(2d_0)+2j}\cdot y^*\left((\nabla_{\partial_{x^j}}F)(\partial_{\xi^a},\partial_{\xi^b})\right)
&=\left((y^*\nabla)_{\partial_{\eta^{d_1+n(2d_0)+2j-1}}}F_y\right)(y^*\partial_{\xi^a},y^*\partial_{\xi^b})\nonumber
\end{align}
thus generalising (21) of \cite{Gro14}.

The polynomial of odd generators in the second line of
(\ref{eqnGalaevYGeneralised}) can be understood as follows.
A product of two $\eta$'s, both different from those occurring in the first line,
is necessary for having a formula like (\ref{eqnGalaevTermsGeneralised}).
In the proof of Lem. \ref{lemHolk2} to be stated next, we need to express
a variant of the left hand side of this equation
containing concatenations of covariant derivatives in the shape of the right
hand side.
The product of two $\eta$'s corresponding to different even coordinates $x^j$ must,
therefore, have nonvanishing product.
In addition, in order to differentiate
with respect to the same direction $\partial_{x^j}$ up to $k$ times,
each $y^{\sharp}(x^j)$ is defined as the sum of $k$ such pairs.

\begin{Lem}
\label{lemHolk2}
Let $q:S\rightarrow M$ be an $S$-point, $k\in\bN$ and $y:S\times T\rightarrow M$
be over $q$ as defined in (\ref{eqnGalaevYGeneralised}).
Then every operator of the form
\begin{align*}
\left(\prod_{n=0}^{k-1}\prod_{j=1}^{d_0}\eta^{d_1+n(2d_0)+2j}\right)\cdot
(\nabla^k_{\partial_{\xi^{j_k}},\ldots,\partial_{\xi^{j_1}}}F)_y
\end{align*}
can be written as an $\mO_{S\times T}$-linar combination of operators of the form
\begin{align*}
(y^*\nabla)_{\partial_{\eta^{i_l}}}\circ\ldots\circ(y^*\nabla)_{\partial_{\eta^{i_1}}}F_y
\end{align*}
for $l\leq k$.
\end{Lem}

\begin{proof}
As in the proof of Lem. \ref{lemHolk}, we proceed by induction.
Again, the base case $k=0$ is trivial. The inductive step follows
from applying (\ref{eqnGalaevTermsGeneralised}), with $F$ replaced
by $\nabla^k_{\partial_{\xi^{j_k}},\ldots,\partial_{\xi^{j_1}}}F$,
and Leibniz' rule to the right hand side of the recursive definition of
Def. \ref{defHigherCovariantDerivatives}.
The construction is such that
conversion of a covariant derivative along an even direction $\partial_{x^j}$
via (\ref{eqnGalaevTermsGeneralised}) swallows one of $k$ available
$\mO_T$-generators associated to $x^j$.
\end{proof}

Prp. 2 of \cite{Gro14} provides a formula for the derivative of
the parallel transport operator by an even homotopy variable $s$.
There, it was assumed that the homotopy is proper, i.e. the boundary
points do not depend on $s$. In general, one would get boundary terms.
The following proposition is the corresponding statement for a single path
$\gamma$ interpreted as a homotopy with respect to one of the odd variables.

\begin{Prp}
\label{prpEtaPgamma}
Let $U=\bR^{0|M}=(*,\langle\eta^1,\ldots,\eta^M\rangle)$ be
a superpoint, $x,y:U\rightarrow M$ be $U$-points and $\gamma:x\rightarrow y$ 
be a $U$-path. Then
\begin{align*}
\partial_{\eta^i}P_{\gamma}[Z_x]
=\left(\int_0^1R_t[\partial_{\eta^i}]\right)P_{\gamma}
&+P_{\gamma}\,\partial_{\eta^i}(x^*(\xi^l))x^*(\nabla_{\partial_{\xi^l}}T^k)\cdot Z_x^k\\
\qquad&-\partial_{\eta^i}(y^*(\xi^l))y^*(\nabla_{\partial_{\xi^l}}T^k)\cdot(P_{\gamma}[Z_x])^k
\end{align*}
for every $Z_x\in x^*\mE$,
where $R_t[\partial_{\eta^i}]:=P_t\,\ev|_t\scal[R_{\gamma}]{d\gamma[\partial_t]}{d\gamma[\partial_{\eta^i}]}P_t^{-1}$
and $P_t:=P_{{\gamma}|_{U\times[t,1]}}$.
\end{Prp}

The formulation of Prp. \ref{prpEtaPgamma} contains some heavy abuse of notation
that needs to be explained. Recall that parallel transport $P_{\gamma}$
is an operator $x^*\mE\rightarrow y^*\mE$. The derivative with respect to
$\eta^i$ thus depends on the choice of trivialisations on either side.
To be precise, we let $(T^k)$ denote a local $\mE$-basis in a neighbourhood
of $x_0\in M_{\overline{0}}$ and expand $Z_x=:(x^*T^k)\cdot Z_x^k$,
and analogous with respect to $y$. The reader should note that we use the
same symbol $(T^k)$ for two different local $\mE$-bases.

\begin{proof}
This is shown along the lines of the proof of Prp. 2 of \cite{Gro14}.
\end{proof}

We will use Prp. \ref{prpEtaPgamma} for $U=S\times T$ and $x$
replaced by $x_T:S\times T\rightarrow M$. In the following, the symbols $\eta^i$
will again be used to denote odd generators of $\mO_T$, as in the beginning of
this section.
By the following corollary, an $\eta_i$-derivative of an Ambrose-Singer operator
yields a covariant derivative together with some lower order term.
It is the key formula for the proof of Thm. \ref{thmGalaevFunctor} below.

\begin{Cor}
\label{corPartialEtaCurvature}
Let $x:S\rightarrow M$ be an $S$-point, considered as an $S\times T$-point
$x=x_T:S\times T\rightarrow M$, $y$ be an $S\times T$-point and
$\gamma:x\rightarrow y$.
Let $F$ be a tensor of curvature type as above and $u,v\in y^*\mT M$. Then
\begin{align*}
&\partial_{\eta^i}(P_{\gamma}^{-1}\circ F_y(u,v)\circ P_{\gamma})\\
&\qquad=P_{\gamma}^{-1}\circ((y^*\nabla)_{\partial_{\eta^i}}F_y)(u,v)\circ P_{\gamma}
+P_{\gamma}^{-1}F_y((y^*\nabla)_{\partial_{\eta^i}}u,v)P_{\gamma}\\
&\qquad\qquad+(-1)^{\abs{u}}P_{\gamma}^{-1}F_y(u,(y^*\nabla)_{\partial_{\eta^i}}v)P_{\gamma}
-\scal[[]{\int_0^1dt\,P_{\gamma}^{-1}R_t[\partial_{\eta^i}]P_{\gamma}}{P_{\gamma}^{-1}F_y(u,v)P_{\gamma}}
\end{align*}
\end{Cor}

\begin{proof}
By the product rule and the standard formula
\begin{align*}
\partial_{\eta^i}P_{\gamma}^{-1}=-P_{\gamma}^{-1}(\partial_{\eta^i}P_{\gamma})P_{\gamma}^{-1}\,,
\end{align*}
the left hand side of the equation can be expressed in the form
\begin{align*}
\partial_{\eta^i}(P_{\gamma}^{-1}\circ F_y(u,v)\circ P_{\gamma})
=-\scal[[]{P_{\gamma}^{-1}\partial_{\eta^i}P_{\gamma}}{P_{\gamma}^{-1}F_y(u,v)P_{\gamma}}+P_{\gamma}^{-1}\scal[[]{\partial_{\eta^i}}{F_y(u,v)}\circ P_{\gamma}
\end{align*}
Application of Prp. \ref{prpEtaPgamma} with the assumption
$\partial_{\eta^i}x^*=0$ then leads to the right hand side.
\end{proof}

To rephrase the previous result in a more convenient form,
we introduce the following definition.

\begin{Def}
\label{defHolk}
For $k\geq 0$, define $\hol_x(T)^{(k)}$ to be the Lie algebra generated by operators
\begin{align*}
P_{\gamma}^{-1}\circ\left((y^*\nabla)_{\partial_{\eta^{i_l}}}\circ\ldots\circ(y^*\nabla)_{\partial_{\eta^{i_1}}}R_y\right)(u,v)\circ P_{\gamma}
\end{align*}
with $\gamma:x\rightarrow y$, $0\leq l\leq k$ and $u,v\in y^*\mT M$.
\end{Def}

This is such that $\hol_x(T)^{(k)}\subseteq\hol_x(T)^{(k+1)}$.
Note that, in general, $\hol_x(T)^{(0)}\neq\hol_x(T)$ since the latter
allows only even $u$ and $v$.

\begin{Cor}
\label{corEtaDerivatives}
Let $x:S\rightarrow M$ be an $S$-point, $y$ be an $S\times T$-point,
$\gamma:x\rightarrow y$ and $u,v\in y^*\mT M$. Then
\begin{align*}
\partial_{\eta^{i_k}}\ldots\partial_{\eta^{i_1}}(P_{\gamma}^{-1}\circ R_y(u,v)\circ P_{\gamma})
=P_{\gamma}^{-1}\circ((y^*\nabla)_{\partial_{\eta^{i_k}}}\ldots(y^*\nabla)_{\partial_{\eta^{i_1}}}R_y)(u,v)\circ P_{\gamma}
+R^{k-1}
\end{align*}
with $R^{k-1}\in\hol_x(T)^{(k-1)}$.
Moreover, $\hol_x(T)^{(m)}$ is generated by operators of the form of the
left hand side of this equation with $0\leq k\leq m$.
\end{Cor}

\begin{proof}
The first part of the statement is shown by induction, with both the
base case $k=1$ and the inductive step provided by
Cor. \ref{corPartialEtaCurvature}.
The second part is a direct corollary of the first.
\end{proof}

\begin{proof}[Proof of Thm. \ref{thmGalaevFunctor}]
$\hol_x^{\mathrm{C}}$ is generated by operators $X$ of the form
\begin{align}
\label{eqnHolC}
X=Y|_{\eta=0}\;,\qquad Y=\partial_{\eta^{i_k}}\ldots\partial_{\eta^{i_1}}(P_{\gamma}^{-1}\circ R_y(u,v)\circ P_{\gamma})
\end{align}
with $u,v\in(y^*\mT M)_{\overline{0}}$.
By Cor. \ref{corEtaDerivatives}, $Y$ is contained in $\hol_x(T)^{(k)}$.
As such, by Lem. \ref{lemHolk} it has an expression as a linear combination of
(possibly nested) commutators involving operators of the form
\begin{align*}
P_{\gamma}^{-1}\circ\left(\nabla^l_{\partial_{\xi^{j_l}},\ldots,\partial_{\xi^{j_1}}}R\right)_y(u,v)\circ P_{\gamma}
\end{align*}
with $0\leq l\leq k$ and $u,v\in y^*\mT M$.
It follows that $X$ has a corresponding expression with the occurrences of
$y$, $u$, $v$ and $\gamma$ replaced by $y_0:=y|_{\eta=0}:S\rightarrow M$
and analogous for $u_0$, $v_0$ and $\gamma_0$, respectively.
By definition (Def. \ref{defGalaevAlgebra}), this is clearly contained in
$\hol_x^{\mathrm{Gal}}$. Galaev's holonomy algebra thus contains all generators
of the holonomy algebra of coefficients, thus establishing the inclusion
$\hol_x^{\mathrm{C}}\subseteq\hol_x^{\mathrm{Gal}}$.

Conversely, consider a generator $X$ of $\hol_x^{\mathrm{Gal}}$ of the form
\begin{align*}
X=P_{\gamma_0}^{-1}\circ\left(\nabla^k_{\partial_{\xi^{i_k}},\ldots,\partial_{\xi^{i_1}}}R\right)_{y_0}(u_0,v_0)\circ P_{\gamma_0}
\end{align*}
with $y_0:S\rightarrow M$ and $\gamma_0:x\rightarrow y_0$ and $u_0,v_0\in y_0^*\mT M$.
Let $y:S\times T\rightarrow M$ (with $T$ sufficiently large) be as
in (\ref{eqnGalaevYGeneralised}),
and $u,v\in y^*\mT M$ such that $u_0=u|_{\eta=0}$ and $v_0=v|_{\eta=0}$. Then
\begin{align*}
X=Y|_{\eta=0}\;,\qquad Y:=P_{\gamma}^{-1}\circ\left(\nabla^k_{\partial_{\xi^{i_k}},\ldots,\partial_{\xi^{i_1}}}R\right)_y(u,v)\circ P_{\gamma}
\end{align*}
In order to apply Lem. \ref{lemHolk2}, we denote the product of $\eta$'s occurring
in that lemma by $\eta^{(d,k)}$ and rewrite the previous equation in the form
\begin{align*}
X=\hat{Y}|_{\eta=0}\;,\qquad\hat{Y}=\partial_{\eta^{(d,k)}}\left(\eta^{(d,k)}Y\right)
\end{align*}
By Lem. \ref{lemHolk2}, the term in parentheses is an $\mO_{S\times T}$-linear
combination of operators as in Def. \ref{defHolk} ($\bR$-linear upon
redefining the vectors $u$ or $v$), and thus contained in $\hol_x(T)^{(k)}$.
By Cor. \ref{corEtaDerivatives}, it is a linear combination
of arbitrary commutators of operators as on the left hand side of the
corollary, and the same follows for $\hat{Y}$.
Therefore, $X$ has a corresponding expression in terms of
operators of the form (\ref{eqnHolC}).
Consider one such operator. We may assume, without loss of generality, that
$u$ (and similarly $v$) is even, for if not, we enlarge $T$ by an additional
odd generator $\eta^a$, replace (the odd part of) $u$ by
$\eta^a\cdot u$ and put a derivative
$\partial_{\eta^a}$ in front.
We conclude that $X$ is contained in $\hol_x^{\mathrm{C}}$ which,
therefore, contains all generators of $\hol_x^{\mathrm{Gal}}$.
The theorem is proved.
\end{proof}

\subsection{Sheaf Properties}
\label{subsecSheafProperties}

Graßmann algebras are special instances of supercommutative superalgebras.
For the study of functors in the latter category, it is natural to consider
the fppf Grothendieck topology which, by definition, is the collection of
finitely many morphisms $R\rightarrow R_i$ such that $R_i$ is finitely
presented and $R_1\times\ldots\times R_n$ is a faithfully flat $R$-module.
An important application concerns the quotient of an algebraic supergroup
by a closed subsupergroup.
The quotient functor, defined by the quotients of the functors of points,
does not come with good properties. Only its fppf-sheafification is
represented by a superscheme as shown in \cite{MZ11}, cf. also Sec. 5.2
in \cite{Jan87}, and \cite{Vis05} for a general treatment on Grothendieck
topologies and related concepts.

By analogy and in view of the results of the previous subsection,
one might conjecture that (the functor of points of)
$\Hol_x^{\mathrm{Gal}}$ should be the sheafification of the functor
$\Hol_x$.
However, this turns out not to be the case. To the contrary,
we shall prove that both functors are in fact sheaves
with respect to the fppf topology (see Prp. \ref{prpHolSheaves} below).
We consider the category $\mP$ of superpoints $\bR^{0|L}$.

\begin{Def}[fppf]
\label{defFPPF}
The fppf topology, denoted $T_{fppf}$, on $\mP$ is defined as collection
of finite sets $\{P_i\rightarrow P\}_{i\in I}$ (for $P,P_i\in\mP$) such that
each morphism $P_i\rightarrow P$ is a submersion.
\end{Def}

Def. \ref{defFPPF} is convenient for the present purposes.
In App. \ref{secFPPF}, we will prove its equivalence with the one
used in \cite{MZ11}, upon restriction to Graßmann algebras,
a result of independent interest.
From that equivalence, one can deduce that, indeed, $T_{fppf}$ satisfies the
axioms of a Grothendieck topology.
We shall now sketch a proof of this result in the current setting, not only
for having a self-contained exposition but, more importantly, to collect
some notation and facts for reference in the proof of Prp. \ref{prpHolSheaves}
below.

Consider morphisms $\varphi_i:\bR^{0|L_i}\rightarrow\bR^{0|L}$ ($i=1,2$).
The fibred product, if it exists, is defined as the object which makes
the following type of diagrams commute and is universal in this respect
(see Chp. III.4 of \cite{ML98}).
\begin{align*}
\begin{xy}
\xymatrix{
\ar@{}[drr]|{\circlearrowleft}\bR^{0|L_1}\times_{\bR^{0|L}}\bR^{0|L_2}
\ar[d]_{\pr_2}\ar[rr]^{\pr_1}&&\bR^{0|L_1}\ar[d]^{\varphi_1}\\
\bR^{0|L_2}\ar[rr]_{\varphi_2}&&\bR^{0|L}}
\end{xy}
\end{align*}
We need the existence of the fibred product in case $\varphi_1$ is a submersion.
This is indeed the case, for a submersion is
transversal to any morphism $\varphi_2$, see Prp. 2.9 of \cite{BBHP98}.
Moreover, the fibred product has the form
\begin{align*}
\bR^{0|L_1}\times_{\bR^{0|L}}\bR^{0|L_2}
=\left(*,\bigwedge\bR^{L_1+L_2}/\left\langle(\varphi_1\times 0-0\times\varphi_2)^*(\bigwedge\bR^L)\right\rangle\right)
\end{align*}
and the maps $\pr_i$ are indeed projections.

Coordinates of this space can be found as follows. As $\varphi_1$ is
assumed to be a submersion, there exist coordinates
$(\eta^1,\ldots,\eta^L,\eta^{L+1},\ldots,\eta^{L_1})$ of
$\bR^{0|L_1}$ such that $\varphi_1$ identifies $\eta^1,\ldots,\eta^L$
with coordinates of $\bR^{0|L}$. Let, moreover,
$(\theta^1,\ldots,\theta^{L_2})$ be any coordinates of $\bR^{0|L_2}$.
It follows that
\begin{align*}
\langle\theta^1,\ldots,\theta^{L_2},\eta^{L+1},\ldots,\eta^{L_1}\rangle
\cong\bigwedge\bR^{L_1+L_2}/\left\langle(\varphi_1\times 0-0\times\varphi_2)^*(\bigwedge\bR^L)\right\rangle
\end{align*}
It follows that the projection $\pr_2$ is a submersion, and we have established
the following lemma.

\begin{Lem}
\label{lemFPPF}
$T_{fppf}$ satisfies the axioms of a Grothendieck
topology, such that $(\mP,T_{fppf})$ constitutes a site.
\end{Lem}

For our next result, recall the definition of the supermanifold
$\fgl_{n|m}^S$ for a fixed Graßmann algebra $S$
and the characterisation of its functor of points in Lem. \ref{lemGLS}.

\begin{Lem}
Every subfunctor $F$ of (the functor of points of)
$\fgl_{n|m}^S$ is a sheaf in $(\mP,T_{fppf})$.
\end{Lem}

\begin{proof}
Considering two submersions $\varphi_i:\bR^{0|L_i}\rightarrow\bR^{0|L}$ ($i=1,2$),
let $(\eta_{(i)}^k)_{1\leq k\leq L_i}$ denote coordinates of
$\bR^{0|L_i}$ such that $\varphi_i$ identifies the first $L$
with coordinates of $\bR^{0|L}$.
Let $a_i\in F(\bR^{0|L_i})$. The condition
$\pr_1^*a_1=\pr_2^*a_2\in F(\bR^{0|L_1}\times_{\bR^{0|L}}\bR^{0|L_2})$.
implies that either side depends only on the respective first $L$ coordinates,
which are identified in the fibred product.
In other words, there is a unique $a\in F(\bR^{0|L})$ such that
$a_i=\varphi_i^*a$.
Repeating the argument finitely many times, the sheaf property is established.
\end{proof}

By Lem. \ref{lemSubfunctors}, the holonomy functors are all subfunctors
of $\fgl_{\rk\mE}^S$. As such, they may be considered for the preceding lemma.
We thus yield the following result.

\begin{Prp}
\label{prpHolSheaves}
The functor of points $T\mapsto\Hol_x^{\mathrm{Gal}}(T)$, as well as the
functor $T\mapsto\Hol_x(T)$ are both sheaves in $(\mP,T_{fppf})$.
\end{Prp}

\section{Applications of the Holonomy Functor}
\label{secApplications}

Semi-Riemannian $S$-supermanifolds are relevant in the context
of supergravities \cite{SS12}.
The main purpose of this section is to establish the de Rham-Wu decomposition
theorem, Thm. \ref{thmIntroDeRhamWu} in the precise shape of Thm. \ref{thmDeRhamWu},
generalising Thm. 11.1 of \cite{Gal09} which,
in turn, is a supergeometric version of the corresponding
theorem proved in \cite{Wu67}.
We begin with Thm. \ref{thmParallelSubsheaves} (Thm. \ref{thmIntroParallel})
on parallel subbundles,
which generalises Thm. 6.1 of \cite{Gal09} and will be needed in the proof.

Recall that $\mE_S$ is locally free as an $\mO_{S\times M}$-supermodule.
An $\mO_{S\times M}$-subsheaf $\mF_S$ is called locally direct if, locally,
$\mE_S$ possesses a basis such that a subbasis thereof spans $\mF_S$.
By the next lemma, this condition is equivalent to
$\mF_S$ being locally free.
We provide an elementary proof but remark
that a variant thereof uses Nakayama's lemma (Prp. I.1.1 of \cite{BBHR91}).
The reader should note that corresponding staments for general
(super-)modules are, in general, false.

\begin{Lem}
\label{lemDirectFree}
Let $\mE\rightarrow M$ be a super vector bundle over a supermanifold
and $\mF\subseteq\mE$ a subsheaf of $\mO_M$-supermodules. Then $\mF$
is locally free if and only if it is locally direct, i.e. if and only
if $\mE(U)$ for $U\subseteq M_{\overline{0}}$ sufficiently small
possesses a basis $(e_A)_A$ of which a subbasis spans $\mF(U)$.
\end{Lem}

\begin{proof}
Assume that $\mF(U)$ is free over $\mO_M(U)$. Without loss of generality,
we may assume that $\mO_M(U)\cong C^{\infty}(U)\otimes\bigwedge\bR^m$
through a choice of coordinates on $M$ and, similarly,
$\mE(U)$ is trivial. Denoting the corresponding adapted basis
by $(e_A)_{1\leq A\leq k+l}$, and some basis of $\mF(U)$ by
$(f_{A'})_{A'}$ (where $A'$ runs through $1\leq A'\leq k'$
and $k+1\leq A'\leq k+l'$), there are coefficients
$C^A_{A'}\in\mO_M(U)$ such that $f_{A'}=\sum_A C^A_{A'}\cdot e_A$.
Let $\tilde{f}_{A'}\in C^{\infty}(U,\bR^{k'+l'})$
denote the corresponding sections obtained
by projecting the $C^A_{A'}$ to $C^{\infty}(U)$. They are linearly independent
over $C^{\infty}(U)$, so there are $(\tilde{e}_{\hat{A}})_{\hat{A}}$
with $k'+1\leq\hat{A}\leq k$ or $k+l'+1\leq\hat{A}\leq k+l$ such that
$(\tilde{f}_{A'},\tilde{e}_{\hat{A}})$ forms a basis of
$C^{\infty}(U,\bR^{k'+l'})$. By construction, the tuple
$(f_{A'},\tilde{e}_{\hat{A}})$ is linearly independent over $\mO_M(U)$
and, therefore, constitutes a basis of $\mE(U)$. $\mF(U)\subseteq\mE(U)$ is direct.
The other implication is trivial.
\end{proof}

\begin{Thm}
\label{thmParallelSubsheaves}
Let $M$ be connected, $\nabla$ be an $S$-connection on $\mE_S$,
$x:S\rightarrow M$ be an $S$-point and $T=\bR^{0|L'}$ with $L'$ sufficiently
large. Then the following holds true.
\begin{enumerate}
\renewcommand{\labelenumi}{(\roman{enumi})}
\item Let $\mF_S\subseteq\mE_S$ be a locally free subbundle of $\mE_S$
which is parallel, i.e. such that $\nabla_YX\in\mF_S$ for all $Y\in\mT M_S$ and $X\in\mF_S$.
Then $\mF_x:=\hat{x}^*\mF_S$ is holonomy invariant $\Hol_x(T)\cdot\mF_x\subseteq\mF_x$.
\item Conversely, let $\mF_x$ be a free submodule of $\mE_x:=x^*\mE$
such that $\Hol_x(T)\cdot\mF_x\subseteq\mF_x$.
Then there exists a unique locally free
subbundle $\mF_S\subseteq\mE_S$ with $\hat{x}^*\mF_S=\mF_x$, which is parallel.
\end{enumerate}
\end{Thm}

\begin{proof}
The proof is more involved, yet in its structure reminiscent of
the one of the holonomy principle (Thm. 2 in \cite{Gro14}) and, moreover,
in parts similar to the proof of Thm. 7.1 in \cite{Gal09}.
Therefore, we consider it sufficient to sketch the relevant steps
of the second assertion (which is the hard one) and
leave the details to the attentive reader.
Concerning notation, we shall follow the aforementioned proof
and treat standard generators of both $S$ and $T$ on an equal footing,
writing $\eta^k$ with an appropriate index $k$.

Without loss of generality, we may assume that $M\cong\bR^{n|m}$
has global coordinates $\xi=(x,\theta)$ and $\mE_S\cong E\otimes\mO_M\otimes\mO_S$
is trivial.
Let $((e_A)_x)_A$ denote a basis
of $\mE_x$ with the implicit understanding that $A$ runs through
a finite index set. Following Galaev's notation, we shall write
$((e_A)_x)_A=(((e_{\oA})_x)_{\oA},((e_{\hat{A}})_x)_{\hat{A}})$ such that
$((e_{\oA})_x)_{\oA}$ is a basis of $\mF_x$.
Let $(e_A)_{x_0}$ denote the projection
of $(e_A)_x$ to the super vector space $x_0^*\mE$, removing the nilpotent part
coming from $\mO_S=\langle\eta^1,\ldots,\eta^L\rangle$. This gives rise to a basis
of $x_0^*\mE$ and is such that every $(e_A)_x$ can be written as a an $\mO_S$-linear
combination of the $(e_A)_{x_0}$. For the $\mF_x$-part, we write
\begin{align*}
(e_{\oA})_x=(e_{\oA})_{x_0}+(T_{\oA}^{\hat{B}})_x\cdot(e_{\hat{B}})_{x_0}\;,\qquad
(T_{\oA}^{\hat{B}})_x\in\mO(\eta)\subseteq\mO_S
\end{align*}
This is not the most general expansion possible, but may be achieved by
an invertible transformation of the $(e_{\oA})_x$.

Now define $F_{x_0}:=\spann_{\bR}\left((e_{\oA})_{x_0}\right)$. By holonomy invariance
of $\mF_x$, we conclude $(\Hol_{x_0})_{\overline{0}}\cdot F_{x_0}\subseteq F_{x_0}$.
By the classical analogon of the current theorem,
this gives rise to a parallel subbundle $F\subseteq E$.
Now let $(e_A)$ be a basis of $E$ such that $(e_{\oA})$ is a basis of $F$
and $x_0^*e_A=(e_A)_{x_0}$. By the assumption of triviality, we may consider
$(e_A)$ as a basis of $\mE_S$ such that $x^*e_A=(e_A)_{x_0}$.
We want a basis $(f_{\oA})$ of $\mF_S$ such that $e_{\oA}$ is the canonical
projection of $f_{\oA}$ and $x^*f_{\oA}=(e_{\oA})_x$.
We find that any
\begin{align*}
f_{\oA}=e_{\oA}+T_{\oA}^{\hat{B}}\cdot e_{\hat{B}}\quad\mathrm{with}\quad
T_{\oA}^{\hat{B}}\in\mO_{S\times M}\quad\mathrm{s.th.}\quad
x^*(T_{\oA}^{\hat{B}})=(T_{\oA}^{\hat{B}})_x
\end{align*}
satisfies these conditions.
We will prove existence and uniqueness of such
$T_{\oA}^{\hat{B}}$ with the property that there exist
functions $X_{a\oA}^{\oB}\in\mO_{S\times M}$ such that
\begin{align*}
\nabla_{\partial_{\xi^a}}f_{\oA}=X_{a\oA}^{\oB}\cdot f_{\oB}
\end{align*}
Then the subsheaf
\begin{align*}
\mF_S:=\spann_{\mO_{S\times M}}(f_{\oA})
=\mO_{S\times M}\otimes\spann_{\bR}(f_{\oA})
\end{align*}
is parallel and locally free.

In the first step, we construct $(T_{\oA}^{\hat{B}})^0\in\mO_M$ recursively
with respect to an expansion in monomials $\theta^I$ of odd $M$-variables,
setting $(T_{\oA}^{\hat{B}})^0(q):=0$ and applying the analogon of
(30) in \cite{Gal09} to the projection $\nabla^{\mE}$ of $\nabla$ to a
superconnnection on $\mE$. The bundle
\begin{align*}
\mF_S^0:=(\mF_S)_0:=\mO_M\otimes\spann_{\bR}(f_{\oA}^0)\;,\qquad
f_{\oA}^0:=e_{\oA}+(T_{\oA}^{\hat{B}})^0\cdot e_{\hat{B}}
\end{align*}
satisfies the following.
First, for $y:S\times T\rightarrow M$ and $\gamma:x\rightarrow y$, we obtain
\begin{align*}
P_{\gamma}[\mF_x]|_{\eta^0}=P_{\gamma_0}[F_{x_0}]=F_{y_0}=(y^*\mF_S^0)|_{\eta^0}
\end{align*}
Moreover, the construction implies that
\begin{align*}
(\nabla_{\partial_{\theta^j}}f_{\oA}^0)|_{\eta^0}=\nabla^{\mE}_{\partial_{\theta^j}}f^0_{\oA}\subseteq\mF_S^0\;,\qquad
(\nabla_{\partial_{x^j}}f_{\oA}^0)|_{\eta^0\theta^0}
=\nabla^0_{\partial_{x^j}}e_{\oA}\subseteq F=\mF_S^0|_{\eta^0\theta^0}
\end{align*}

We take the preceding properties as the beginning of an induction
and consider multiindices $I=(i_1,\ldots,i_{\abs{I}})$ with
$1\leq i_j\leq L+(\dim M)_{\overline{1}}$, such that $\eta^I\in\mO_{S\times T}$.
Assume that we have constructed
$(T_{\oA}^{\hat{B}})^N,(X_{a\oA}^{\oB})^N\in\mO_{S\times M}$ and $(M_{\oA}^{\oB})^N_{(\gamma)}\in\mO_{S\times T}$
for $N\in\bN$ such that
\begin{enumerate}
\setcounter{enumi}{-1}
\renewcommand{\labelenumi}{$\arabic{enumi}_N$}
\item $(T_{\oA}^{\hat{B}})^N$ has an expansion $X^N=\sum_{\abs{I}\leq N}X|_{\eta^I}\cdot\eta^I$
such that $X|_{\eta^I}=0$ whenever there is $i_j\in I$ with $i_j\geq L+1$,
and accordingly for $(X_{a\oA}^{\oB})^N$ and $(M_{\oA}^{\oB})^N_{(\gamma)}$.
\item $P_{\gamma}[(e_{\oA})_x]|_{\eta^I}=\left((M_{\oA}^{\oB})^N_{(\gamma)}\cdot y^*f_{\oB}^N\right)|_{\eta^I}$
for every $y:S\times T\rightarrow M$, $\gamma:x\rightarrow y$ and $\abs{I}\leq N$.
\item $(\nabla_{\partial_{\theta^j}}f_{\oA}^N)|_{\eta^I}=\left((X_{\theta^j\oA}^{\oB})^N\cdot f_{\oB}^N\right)|_{\eta^I}$ for all $\abs{I}\leq N$.
\item $(\nabla_{\partial_{x^j}}f_{\oA}^N)|_{\theta^A\eta^B}=\left((X_{x^i\oA}^{\oB})^Nf_{\oB}^N\right)_{\eta^I}$ for all $A$,$B$ such that
$\abs{A}+\abs{B}\leq N$, where $A=(a_1,\ldots,a_{\abs{A}})$ with $1\leq a_j\leq(\dim M)_{\overline{1}}$.
\end{enumerate}

The inductive step proceeds as follows.
For $\abs{J}=N+1$, we define $(M_{\oA}^{\oB})^J$ and $(T_{\oB}^{\hat{B}})^J(q)$
by $1_{N+1}$. While the former explicitly depends on $y$ and $\gamma$,
one can show that the latter depends only on $q\in M_0$ such that $y_0(0)=q$.
By means of $2_{N+1}$ we next endow $(T_{\oB}^{\hat{B}})^J(q)$ to
$(T_{\oB}^{\hat{B}})^J\in\mO_M$ and define $(X_{\theta^j\oA}^{\oB})^J$.
Finally use $3_{N+1}$ to define $(X_{x^i\oA}^{\oB})^{N+1}$.

The induction hypotheses with respect to $L+(\dim M)_{\overline{1}}$
imply that the covariant derivatives of the sections $f_{\oA}$ remain in their span
such that the subsheaf $\mF_S$ is parallel.
It is unique by an argument similar to the last part of the proof
of Thm. 2 of \cite{Gro14}.
\end{proof}

In view of our theory, it is natural to generalise the notion of
a semi-Riemannian supermanifold to the relative setting with respect
to a superpoint $S$, which is straightforward. Considering coefficients
of odd $S$-generators, which can be thought of as auxiliary data on $M$,
the supermetric can be interpreted as having an even and an odd part.
This is analogous to considering ''maps with flesh'' $S\times M\rightarrow N$
as models for superfields including bosons (even) and fermions (odd),
see e.g. \cite{Hel09,DF99b,Khe07,Han12,Gro11a} for this concept under various
names, and especially \cite{SS12} where the following generalisation of a
semi-Riemannian supermanifold occurs in the context of supergravities.

\begin{Def}
A semi-Riemannian $S$-supermanifold is a supermanifold $M$ together with
an $S$-metric, i.e. an even non-degenerate, supersymmetric, bilinear form
$g\in\Hom_{\mO_{S\times M}}(\mT M_S\otimes_{\mO_{S\times M}}\mT M_S,\mO_{S\times M})$.
\end{Def}

The torsion tensor of an $S$-connection is defined as usual.
\begin{align*}
T(X,Y):=\nabla_XY-(-1)^{\abs{X}\abs{Y}}\nabla_YX-\scal[[]{X}{Y}\,,\quad
T\in\Hom_{\mO_{S\times M}}(\mT M_S\otimes\mT M_S,\mT M_S)_{\overline{0}}
\end{align*}

\begin{Lem}
Let $(M,g)$ be a semi-Riemannian $S$-supermanifold. It possesses a unique
$S$-connection, denoted $\nabla^{LC}:\mT M_S\rightarrow\mT M_S^*\otimes\mO_{S\times M}\mT M_S$, such that $T=0$ and it is metric:
\begin{align*}
X\scal[g]{Y}{Z}=\scal[g]{\nabla_XY}{Z}+(-1)^{\abs{X}\abs{Y}}\scal[g]{Y}{\nabla_XZ}
\end{align*}
\end{Lem}

\begin{proof}
Shown as usual via Koszul's formula.
\end{proof}

\begin{Lem}
\label{lemMetricPullback}
Let $(N,g)$ be a semi-Riemannian $S$-supermanifold, $M$ a supermanifold,
and $\varphi:S\times M\rightarrow N$ a morphism.
Let $\nabla$ an $S$-superconnection on $N$ which is metric (such as Levi-Civita).
Then the pullback $\varphi^*\nabla:\varphi^*\mT N\rightarrow\mT M_S^*\otimes_{\mO_{S\times M}}\varphi^*\mT N$ is metric in the following sense.
\begin{align*}
X\scal[g_{\varphi}]{Y}{Z}=\scal[g_{\varphi}]{(\varphi^*\nabla)_X Y}{Z}+(-1)^{\abs{X}\abs{Y}}\scal[g_{\varphi}]{Y}{(\varphi^*\nabla)_X Z}
\end{align*}
holds true for every $X\in\mT M$ as well as $Y,Z\in\varphi^*\mT N$.
\end{Lem}

\begin{proof}
This follows from a straightforward calculation in local coordinates.
\end{proof}

Let now $S=\bR^{0|L}$ and $(M,g)$ be a semi-Riemannian $S$-supermanifold
with vector bundle $\mE=\mT M$. Let $x:S\rightarrow M$ be an $S$-point.
We call a free $\mO_S$-submodule $W\subseteq\mE_x=x^*\mE$ non-degenerate, if it
is with respect to the pullback metrix $g_x$, in other words if
$g_x|_W$ is non-degenerate. From Lem. \ref{lemDirectFree}, we know
that $W$ is direct. As in the classical case,
there is a canonical choice of complementing submodule as shown next.

\begin{Lem}
\label{lemOrthogonalComplement}
Let $W\subseteq\mE_x$ be a free non-degenerate submodule.
Then its ortogonal complement
\begin{align*}
W^{\perp}:=\{v\in\mE_x\setsep\scal[g_x]{v}{W}=0\}\subseteq\mE_x
\end{align*}
is a free and non-degenerate submodule of $\mE_x$ such that
$\mE_x=W\oplus W^{\perp}$.
\end{Lem}

\begin{proof}
$g_x|_W$ is even, nondegenerate and supersymmetric, thus $W$ possesses
an $OSp$-basis $(e_1,\ldots,e_k,f_1,\ldots,f_{2l})$. Continuing performing the
corresponding algorithm, we may endow this basis to an $OSp$-basis
$(e_1,\ldots,e_n,f_1,\ldots,f_{2m})$ of $\mE_x$. It is then clear that
$(e_{k+1},\ldots,f_{2l+1},\ldots)$ is an $OSp$-basis of $W^{\perp}$.
\end{proof}

\begin{Def}
Let $x:S\rightarrow M$.
The holonomy group $\Hol_x(T)$, for fixed $T$,
is called \emph{strongly reducible} if
there is a free non-degenerate submodule $F\subseteq\mE_x$ which is
preserved $\Hol_x(T)\cdot F\subseteq F$.
Otherwise, we call it \emph{weakly irreducible}.

The holonomy group functor $T\mapsto\Hol_x(T)$ is called
\emph{strongly reducible}
if there is free non-degenerate submodule $F\subseteq\mE_x$ preserved
by $\Hol_x(T)$ for all superpoints $T$.
\end{Def}

By Thm. \ref{thmTwofold}, the holonomy group functor is weakly irreducible
if and only if $\Hol_x(T)$ is weakly irreducible for one fixed
sufficiently large $T$. Moreover, this property depends
on the connected component of $x$ only:

\begin{Lem}
Let $x,y:S\rightarrow M$ be $S$-points connected by an $S\times T$-path
$\gamma$. Parallel transport $P_{\gamma}$
identifies free non-degenerate submodules of $\mE_x$ and $\mE_y$,
preserved by $\Hol_x(T)$ and $\Hol_y(T)$, respectively.
\end{Lem}

\begin{proof}
This follows from the following two important observations.
First, the holonomy groups are conjugated by parallel transport.
\begin{align*}
\Hol_x(T)=P_{\gamma}^{-1}\circ\Hol_y(T)\circ P_{\gamma}
\end{align*}
Second, by Lem. \ref{lemMetricPullback} applied to $X=\partial_t$, parallel
transport is an isometry.
\end{proof}

Let $x:S\rightarrow M$ and assume that $(M,g)\cong(M_1\times M_2,g_1+g_2)$
splits. We define $x_i:=\pr_i\circ x:S\rightarrow M_i$. It is clear that
$\mT M_x\cong\mT(M_1)_{x_1}\oplus\mT(M_2)_{x_2}$. Likewise, parallel transport
splits, and $\Hol^{\nabla^g}_x(T)\cong\Hol^{\nabla^{g_1}}_{x_1}(T)\times\Hol^{\nabla^{g_2}}_{x_2}(T)$.
We are now in a position to state a theorem of de Rham-Wu type for
the present case of an semi-Riemannian $S$-supermanifold.
Here, the subscript ''$0$'' does not refer to the underlying manifold
or a related notion, as should be clear from the context.

\begin{Thm}[De Rham-Wu]
\label{thmDeRhamWu}
Let $(M,g)$ be a semi-Riemannian $S$-supermanifold such that
the underlying classical semi-Riemannian manifold is simply connected
and geodesically complete.
Then there exist semi-Riemannian $S$-supermanifolds
$(M_i,g_i)$, $0\leq i\leq r$ with $r\in\bN$, such that
\begin{align*}
(M,g)=(M_0\times M_1\times\ldots\times M_r,g_0+g_1+\ldots+g_r)
\end{align*}
The supermanifold $(M_0,g_0)$ has vanishing curvature
(is flat), and the holonomy group functors
$T\mapsto\Hol^{\nabla^{g_i}}_{x_i}(T)$, $1\leq i\leq r$ are weakly irreducible.
In particular,
\begin{align*}
\Hol_x^{\nabla^g}(T)=\Hol_{x_0}^{\nabla^{g_0}}(T)\times\Hol_{x_1}^{\nabla^{g_1}}(T)\times\ldots\times\Hol_{x_r}^{\nabla^{g_r}}(T)
\end{align*}
for every $S$-point $x:S\rightarrow M$ and $T$ sufficiently large.
\end{Thm}

\begin{proof}
With the theorem on parallel subbundles, Thm. \ref{thmParallelSubsheaves},
established, the proof is similar to Galaev's proof for the case $L=0$
(Thm. 11.1 in \cite{Gal09}).
Assume that $\Hol_x(T)$ is strongly reducible, then it preserves
a free non-degenerate submodule $F_1\subseteq\mE_x$, i.e.
$\Hol_x(T)\cdot F_1\subseteq F_1$.
Let $F_2:=F^{\perp}$ be its orthogonal complement.
By Lem. \ref{lemOrthogonalComplement}, this is a transversal
free non-degenerate submodule.
Parallel transports around loops are isometries, and it follows that
$\Hol_x(T)$ preserves the decomposition $\mE_x=F_1\oplus F_2$.
If $T$ is sufficiently large, we may use Thm. \ref{thmParallelSubsheaves}
to conclude existence and uniqueness of locally direct
parallel subbundles $\mF_1$, $\mF_2$ of $\mT M_S$.
Since the Levi-Civita connection is torsion-free
\begin{align*}
0=\scal[T]{X}{Y}=\nabla_XY-(-1)^{\abs{X}\abs{Y}}\nabla_YX-\scal[[]{X}{Y}
\end{align*}
and the subbundles are parallel, it holds that $\scal[[]{X}{Y}\in\mF_i$ for $X,Y\in\mF_i$
for $i=1,2$, such that the subbundles are involutive.
Let $\tilde{\mF_i}\subseteq\mT M$ denote the canonical projection of
$\mF_i$ to $\mT M$ by setting all generators of $\mO_S$ to zero.
It follows that $\tilde{\mF_i}$ is still
free, the decomposition $\mT M=\tilde{\mF_1}\oplus\tilde{\mF_2}$ is
direct, and the $\mF_i$ are involutive.
By Frobenius' theorem, there are maximal integral subsupermanifolds
$M_1$, $M_2$ with $M_i$ corresponding to $\tilde{\mF_i}$
such that $M$ is locally diffeomorphic to $M_1\times M_2$.
It is globally so by the classical de Rham-Wu theorem.
Moreover, the restrictions $g_i$ of $g$ to
$\mF_i$ are non-degenerate and depend on $M_i$ only.
It follows that $(M,g)\cong(M_1\times M_2,g_1+g_2)$.
\end{proof}

\section*{Acknowledgements}

I would like to thank Jan Hakenberg for lending his algorithmic expertise
and Alexander Alldridge and Dominik Ostermayr for making suggestions
which helped to improve a previous version of the article.

\appendix

\section{The fppf Topology on the Category of Superpoints}
\label{secFPPF}

For the purpose of studying the holonomy group functors it was natural to
define the $fppf$-topology on the category $\mP$ of superpoints $\bR^{0|L}$
in the form of Def. \ref{defFPPF}.
The objective of this appendix is to relate this definition to the
one which is well-known in algebraic geometry under the same name.
For concreteness, we consider the $\bZ_2$-graded version used
in \cite{Zub09, MZ11} of the classical notion (cf. Sec. 5.2 in \cite{Jan87}),
as follows.

Let $\mathrm{SAlg_K}$ denote the category of supercommutative superalgebras
over a field $K$ with characteristic different from $2$.
Recall that a supercommutative superalgebra $B$ over a superring $A$
is said to be finitely presented
if it is of the form $A[t_1,\ldots,t_n|\theta_1,\ldots,\theta_m]/I$
where $I$ is a finitely generated ideal.
An $A$-supermodule $Y$ is said to be faithfully flat if the following holds:
A sequence $X'\rightarrow X\rightarrow X''$ is exact if and only if
$X'\otimes_AY\rightarrow X\otimes_AY\rightarrow X''\otimes_AY$ is exact
(for all $A$-modules $X',X,X''$).

\begin{Def}[fppf]
\label{defFPPFOriginal}
The fppf topology on $\mathrm{SAlg_K}$ is defined as collection of finite
sets $\{R\rightarrow R_i\}_{i\in I}$ (for $R,R_i\in SAlg_K$) such that
the $R$-supermodule $\bigtimes_{i\in I}R_i$ is faithfully flat and
all $R_i$ are finitely presented $R$-superalgebras.
\end{Def}

$\mP$ is naturally equivalent to the category $\mathrm{Gr}$
of Graßmann algebras over $\bR$ and, as such, can be considered as a full
subcategory of $\mathrm{SAlg_{\bR}}$. It makes sense to consider
Def. \ref{defFPPFOriginal} restricted to $\mP$. Doing so results
in our previous notion of fppf topology as shown by the following
main result of this appendix.

\begin{Thm}
\label{thmFPPF}
The topology on $\mP$ as defined in Def. \ref{defFPPF} agrees with
the one induced by Def. \ref{defFPPFOriginal}.
\end{Thm}

\begin{proof}
We note first that being finitely presented is no condition in
the case of Graßmann algebras.
Consider a covering consisting of morphisms $\varphi_i:P_i\rightarrow P$
as in Def. \ref{defFPPF}. By the following proposition,
Prp. \ref{prpFlatSubmersion}, $\varphi_i$ is a submersion
if and only if $R_i=\mO_{P_i}$ is flat as an $R=\mO_P$-module
with respect to $\varphi^*$ (a morphism in $\mathrm{Gr}$).
All $R_i$ being flat, in turn, is equivalent to the condition
of Def. \ref{defFPPFOriginal} (by Lem. I.2.2 of \cite{Bou72}
and the implication $(ii)\implies(i)$ in Prp. \ref{prpFlatSubmersion}).
\end{proof}

Concerning the preceding proof,
we are obviously in the supercommutative rather than the commutative situation
treated in \cite{Bou72}. However, the results relevant for our present purposes
continue to hold unchanged.

In the following, we shall make no notational distinction between a morphism
$\varphi:\bR^{0|L}\rightarrow\bR^{0|L'}$ and its pullback
$\varphi:\bigwedge\bR^{L'}\rightarrow\bigwedge\bR^L$.

\begin{Prp}
\label{prpFlatSubmersion}
Let $\varphi:R\rightarrow S$ be a morphism in $\mathrm{Gr}$,
and consider $S$ as an $R$-module via $\varphi$. Then the following are equivalent.
\begin{enumerate}
\renewcommand{\labelenumi}{(\roman{enumi})}
\item $S$ is faithfully flat.
\item $S$ is flat.
\item $S$ is free.
\item (The associated morphism of superpoints corresponding to) $\varphi$
is a submersion.
\end{enumerate}
\end{Prp}

\begin{proof}
$ $

$(i)\iff(ii)$:
In general, an $R$-module $M$ is faithfully flat if and only if $M$ is flat
and $M\neq mM$ for every maximal ideal $m$ (\cite{Bou72}, Prp. I.3.1).
In the case of a Graßmann algebra $R$, its nilpotent part $R_{\mathrm{nil}}$
is the unique maximal ideal and, obviously, $S\neq\varphi(R_{\mathrm{nil}})S$.

$(ii)\iff(iii)$:
The Jacobson radical of $R$ is $R_{\mathrm{nil}}$, and
$R/R_{\mathrm{nil}}\cong\bR$ is a field.
The equivalence follows from Prp. II.3.5 of \cite{Bou72}.

$(iv)\implies(iii)$:
A submersion is characterised by the existence of coordinates
$(\theta^i)_{1\leq i\leq n}$ and $(\eta^j)_{1\leq j\leq n+m}$
on $R$ and $S$, respectively, such that $\theta^i\mapsto\eta^i$
(see Prp. 5.2.5 of \cite{CCF11}).
Then any $\bR$-basis of $\langle\eta^{n+1},\ldots,\eta^{n+m}\rangle$
is an $R$-basis of $S$.

$(ii)\implies(iv)$:
This is a special case of Prp. 3.6.1(ii) of \cite{Sch89}.
\end{proof}

Along with the proposition, also Thm. \ref{thmFPPF} is established.
We remark that the result by Schmitt in \cite{Sch89}, used for the last
implication, is really the hardest bit of the proof.
While Schmitt's results are stronger than needed here, his proofs
involve a heavy algebraic machinery.
We therefore consider it worth providing an independent, less abstract,
proof of the implication $(iii)\implies(iv)$ in Prp. \ref{prpFlatSubmersion}
in the remainder of this appendix.
It is shown by induction over the number of $R$-Graßmann-generators, while
the base case is established by means of ideal theory and reduction to
a special case considered by Esin and Ko\c{c} in \cite{EK07}.
The proof of the latter result, in turn, is rather concrete.
We start with two easy lemmas needed in the inductive step,
Prp. \ref{prpInductiveStep} below.

\begin{Lem}
Let $\varphi:\bigwedge\bR^n\rightarrow\bigwedge\bR^m$ be a morphism
of Graßmann algebras such that $\bigwedge\bR^m$ is free.
Then $\varphi$ is injective. In particular, $n\leq m$.
\end{Lem}

For the following lemma, we need the inclusion maps
$\pi_j:\bigwedge\bR^{n-1}\rightarrow\bigwedge\bR^n$ with $1\leq j\leq n$,
defined by $\theta^k\mapsto\theta^k$ for $k<j$ and
$\theta^k\mapsto\theta^{k+1}$ for $k\geq j$, where
we let $(\theta^k)_{1\leq k\leq n}$ denote fixed coordinates of $\bigwedge\bR^n$
and analogous for $\bigwedge\bR^{n-1}$.
Unless said otherwise, we consider those
corresponding to the standard bases of $\bR^n$ and $\bR^{n-1}$, respectively.

\begin{Lem}
\label{lemInheritedFreeness}
Let $\varphi:\bigwedge\bR^n\rightarrow\bigwedge\bR^m$ be a morphism
of Graßmann algebras such that $\bigwedge\bR^m$ is free.
Then it is also free with respect to
$\varphi\circ\pi_j:\bigwedge\bR^{n-1}\rightarrow\bigwedge\bR^m$.
\end{Lem}

\begin{Prp}
\label{prpInductiveStep}
Assuming that every morphism $\psi:\bigwedge\bR^n\rightarrow\bigwedge\bR^m$,
such that $\bigwedge\bR^m$ is a free $\bigwedge\bR^n$-module
with respect to $\psi$, is a submersion, the corresponding statement
holds for all morphisms $\varphi:\bigwedge\bR^{n+1}\rightarrow\bigwedge\bR^m$.
\end{Prp}

\begin{proof}
Let $\varphi:\bigwedge\bR^{n+1}\rightarrow\bigwedge\bR^m$ be such that
$\bigwedge\bR^m$ is a free $\bigwedge\bR^{n+1}$-module.
By Lem. \ref{lemInheritedFreeness}, it is also free with respect to the map
$\varphi_{n+1}:=\varphi\circ\pi_{n+1}:\bigwedge\bR^n\rightarrow\bigwedge\bR^m$
which, by assumption, is a submersion.
Therefore, there are coordinates of $\bigwedge\bR^n$ and
$\bigwedge\bR^m$, respectively, still denoted
$\theta^1,\ldots,\theta^n$ and $\eta^1,\ldots,\eta^m$, such that
$\varphi_{n+1}(\theta^i)=\eta^i$ (\cite{CCF11}, Prp. 5.2.5). Endowing the former
coordinates with the original $\theta^{n+1}$, $\varphi$ obtains the form
\begin{align*}
\varphi(\theta^1)=\eta^1,\ldots,\varphi(\theta^n)=\eta^n\,,\qquad
\varphi(\theta^{n+1})\in\bigwedge\bR^m
\end{align*}
We may assume that $\varphi(\theta^{n+1})\in\bigwedge\bR^m\setminus\langle\eta^1,\ldots,\eta^n\rangle$,
for if not we can modify $\theta^{n+1}$ by subtracting from it a suitable element of
$\langle\theta^1,\ldots,\theta^n\rangle$.
Denoting the associated morphism of superpoints still by $\varphi$,
the differential at the single topological point $0$ assumes the form
\begin{align*}
\renewcommand{\arraystretch}{1.6}
(d\varphi)_0
=\left(\arr{cc}{1_{n\times n}&0_{(m-n)\times n}\\0_{1\times n}&\left(\dd[\varphi(\theta^{n+1})]{\eta^{n+1}}|_0\,\ldots\right)}\right)
\end{align*}
$\varphi$ is a submersion if and only if the lower right submatrix is non-zero.
This condition is satisfied by the following argument.
The last line of $(d\varphi)_0$ equals the differential $(d\varphi_{n+1})_0$
of the map $\varphi_{n+1}:=\varphi\circ\pi_1\circ\ldots\circ\pi_n$, where the
$\pi_j$ are defined with respect to the new coordinates.
But $\varphi_{n+1}$ is a submersion by Lem. \ref{lemInheritedFreeness}
and the induction hypothesis for $n=1$.
\end{proof}

We now turn to the base case $n=1$.
The following two lemmas provide equivalent characterisations of freeness
in terms of ideal theory.

\begin{Lem}
\label{lemFreenessBasis}
Let $\varphi:\bigwedge\bR^1\rightarrow\bigwedge\bR^L$ be a morphism
of Graßmann algebras. $\bigwedge\bR^L$ is free
if and only if it has an $\bR$-basis of the form
$(v_1,\ldots,v_{2^{L-1}},\varphi(\theta^1)\cdot v_1,\ldots,\varphi(\theta^1)\cdot v_{2^{L-1}})$.
\end{Lem}

\begin{proof}
This is shown analogous to the proof of Lem. \ref{lemInheritedFreeness}.
\end{proof}

\begin{Lem}
\label{lemFreeness}
Let $\mu\in(\bigwedge\bR^L)_{\overline{1}}$ and consider the ideal
$(\mu)$ in $\bigwedge\bR^L$ generated by $\mu$. Then $\bigwedge\bR^L$
admits an $\bR$-basis of the form
$(v_1,\ldots,v_{2^{L-1}},\mu\cdot v_1,\ldots,\mu\cdot v_{2^{L-1}})$
if and only if $\dim_{\bR}(\mu)\geq 2^{L-1}$.
In this case, $\dim_{\bR}(\mu)=2^{L-1}$.
\end{Lem}

\begin{proof}
If a basis as stated exists, then the vectors $\mu\cdot v_i\in(\mu)$
are all linearly independent, thus the real dimension of $(\mu)$ is greater
than or equal to their number, $2^{L-1}$.

Conversely, let $(\mu\cdot w_1,\ldots,\mu\cdot w_d)$
denote a real basis of $(\mu)$ with $d\geq 2^{L-1}$.
We may endow this basis by vectors $v_j$, $1\leq j\leq f$ to a basis
\begin{align*}
(v_1,\ldots,v_f,\mu\cdot w_1,\ldots,\mu\cdot w_d)
\end{align*}
of $\bigwedge\bR^L$. It follows that $f+d=2^L$ and thus $f\leq 2^{L-1}$.
Multiplying all vectors with $\mu$ and using $\mu^2=0$, one sees that the vectors
$(\mu\cdot v_1,\ldots,\mu\cdot v_f)$ span $(\mu)$,
whence $f\geq\dim(\mu)=d\geq 2^{L-1}$, such that
$f=d=2^{L-1}$. In particular, $(\mu\cdot v_1,\ldots,\mu\cdot v_d)$
is a basis of $(\mu)$. Endowed with the vectors $v_i$, we obtain a basis
of $\bigwedge\bR^L$ as claimed.
\end{proof}

Our strategy for the base case will be to transform $\mu:=\varphi(\theta^1)$
to another odd element of some bigger Graßmann algebra with similar properties,
and such that the associated ideal has the form treated by Esin and Ko\c{c}
in \cite{EK07}. This is Prp. \ref{prpReplacementAlgorithm} below.
We now continue with two lemmas used in the proof of that algorithm.
The first one is clear by Lem. \ref{lemFreeness}.

\begin{Lem}
\label{lemAutomorphismDimension}
Let $\mu\in(\bigwedge\bR^L)_{\overline{1}}$.
If $\dim_{\bR}(\mu)\geq 2^{L-1}$ then $\dim_{\bR}(\varphi(\mu))\geq 2^{L-1}$
for every Graßmann automorphism $\varphi:\bigwedge\bR^L\rightarrow\bigwedge\bR^L$.
\end{Lem}

\begin{Lem}
\label{lemIdealReplacement}
Let $\mu\in(\bigwedge\bR^L)_{\overline{1}}$.
Let $r\in(\bigwedge\bR^L)_{\overline{0}}$ and consider
$\hat{\mu}:=\mu+\eta^{L+1}\cdot r\in(\bigwedge\bR^{L+1})_{\overline{1}}$.
If $\dim_{\bR}(\mu)\geq 2^{L-1}$ (with $(\mu)$ as an ideal in $\bigwedge\bR^L$),
then $\dim_{\bR}(\hat{\mu})\geq 2^L$ (as an ideal in $\bigwedge\bR^{L+1}$).
\end{Lem}

\begin{proof}
Writing elements of $\bigwedge\bR^{L+1}$ in the form $v+\eta^{L+1}w$
with $v,w\in\bigwedge\bR^L$, we obtain
\begin{align*}
(\hat{\mu})
=\{\mu v+\eta^{L+1}(rv-\mu w)\setsep v,w\in\bigwedge\bR^L\}
\end{align*}
By assumption and Lem. \ref{lemFreeness}, $\bigwedge\bR^L$ has a real basis
$(v_1,\ldots,v_{2^{L-1}},\mu\cdot v_1,\ldots,\mu\cdot v_{2^{L-1}})$.
Let $V:=\spann_{\bR}(v_1,\ldots,v_{2^{L-1}})\subseteq\bigwedge\bR^L$.
There is a canonical map from $V\oplus V$ to the space
\begin{align*}
\{\mu v+\eta^{L+1}(rv-\mu w)\setsep v,w\in V\}\subseteq(\hat{\mu})
\end{align*}
sending $(v,w)$ to $\mu v+\eta^{L+1}(rv-\mu w)$, which is
clearly $\bR$-linear and surjective.
Assume $\mu v+\eta^{L+1}(rv+\mu w)=0$.
Then, in particular, $\mu v=0$. By definition of $V$, it follows that
$v=0$. Then also $\mu w=0$ and, similarly, $w=0$.
The aforementioned map is injective, and we conclude that
$\dim_{\bR}(\hat{\mu})\geq\dim(V\oplus V)=2^L$.
\end{proof}

For the next proposition, we need the following notation.
Let $(\eta^i)_i$ denote the coordinates of $\bigwedge\bR^L$ corresponding
to the standard basis of $\bR^L$. We write the expansion of
$\mu\in\bigwedge\bR^L$ with respect to these coordinates in the form
\begin{align}
\label{eqnMuExpansion}
\mu=\sum_JC_{\mu}^J\eta^J\;,\qquad C_{\mu}^J\in\bR\;,\qquad
\eta^J=\eta^{J_1}\cdot\ldots\cdot\eta^{J_{\abs{J}}}
\end{align}
where the sum runs over all multiindices $J$ of length $\abs{J}$ up to $L$.

\begin{Prp}
\label{prpReplacementAlgorithm}
Let $\mu\in(\bigwedge\bR^L)_{\overline{1}}$ be such that
$\dim_{\bR}(\mu)\geq 2^{L-1}$. Then there are $L'\geq L$ and
$\mu'\in(\bigwedge\bR^{L'})_{\overline{1}}$ such that
\begin{itemize}
\item There is a bijective correspondence
\begin{align*}
\lambda:\{J\setsep C_{\mu}^J\neq 0\}\rightarrow\{J'\setsep C_{\mu'}^{J'}\neq 0\}
\end{align*}
such that $\abs{J}=\abs{\lambda(J)}$.
\item The product $\prod_{J'\setsep C_{\mu'}^{J'}\neq 0}\eta^{J'}$ is non-zero.
\item $\dim_{\bR}(\mu')\geq 2^{L'-1}$.
\end{itemize}
\end{Prp}

\begin{proof}
$\mu'$ is successively built from $\mu$ as follows.
Let $j_0$ denote the smallest integer such that the generator
$\eta^{j_0}$ is contained in at least two monomials $\eta^J$ such that
$C_{\mu}^J\neq 0$. Let $I$ denote one of $k\geq 2$ such
multiindices.
By assumption, there is $r\in(\bigwedge\bR^L)_{\overline{0}}$ such that
$\eta^{I}=\eta^{j_0}\cdot r$. Consider
\begin{align*}
\hat{\mu}:=(\mu-C_{\mu}^I\eta^I)+(\eta^{j_0}+\eta^{L+1})\cdot C_{\mu}^Ir
\end{align*}
By Lem. \ref{lemIdealReplacement}, it satisfies
$\dim_{\bR}(\hat{\mu})\geq 2^{L'-1}$ with $L'=L+1$. Consider next the
automorphism $\varphi:\bigwedge\bR^{L'}\rightarrow\bigwedge\bR^{L'}$
defined by $\eta^i\mapsto\eta^i$ ($i<L'$) and
$\eta^{L'}\mapsto(\eta^{L'}-\eta^{j_0})$. Then the element
$\varphi(\hat{\mu})$ satisfies $\dim(\varphi(\hat{\mu}))\geq 2^{L'-1}$
by Lem. \ref{lemAutomorphismDimension}. Moreover, the number of
monomials containing $\eta^{j_0}$ is reduced to $k-1$. It is also
clear that the multiindex bijection required in the statement is satisfied.

Now start with $\mu$ replaced by $\varphi(\hat{\mu})$ from the previous step,
and repeat the construction until finally there is no generator $\eta^{j_0}$
contained in more than one monomial.
Since every step of the construction satisfies the first and third items
in the statement, the same holds for the final result. As no two
monomials therein share a common generator, the product over all
is non-zero.
\end{proof}

\begin{proof}[Proof of $(iii)\implies(iv)$ in Prp. \ref{prpFlatSubmersion}]
This remaining implication is proved by induction over $n$ in
$R=\bigwedge\bR^n$. It remains to show the base case as the inductive step
was already established in Prp. \ref{prpInductiveStep}.
Consider thus a morphism $\varphi:\bigwedge\bR^1\rightarrow\bigwedge\bR^L$
such that $\bR^L$ is free as an $\bR^1$-module via $\varphi$.
By Lem. \ref{lemFreenessBasis} and Lem. \ref{lemFreeness},
this property is characterised by $\dim_{\bR}(\mu)\geq 2^{L-1}$ for
$\mu:=\varphi(\theta^1)$. The algorithm of Prp. \ref{prpReplacementAlgorithm}
constructs another odd $\mu'$ with a similar shape, but such that the product
over all monomials with nonvanishing coefficients in (\ref{eqnMuExpansion})
does not vanish. This is the case treated in \cite{EK07}.
In Thm. 4 of that reference, the dimension of $(\mu')$
is explicitly calculated to be
\begin{align*}
\dim_{\bR}(\mu')=2^{L'-1}\left(1-\prod_{\left\{J\setsep C_{\mu'}^J\neq 0\right\}}\left(1-2^{1-\abs{J}}\right)\right)
\end{align*}
Together with $\dim(\mu')\geq 2^{L'-1}$, this forces at least one of the
multiindices $J$ in the product to be of length $\abs{J}=1$.
But then the corresponding $\mu$-multiindex $\lambda^{-1}(J)$
has also length $1$. It follows that $\varphi$ is a submersion.
\end{proof}

\bibliographystyle{alpha}

\end{document}